\documentclass[10pt,a4paper,english]{amsart}
\usepackage{babel}
\usepackage[numbers]{natbib}
\usepackage{amsfonts, amsmath, wasysym}
\usepackage{amssymb, amsthm}
\usepackage{mathrsfs}
\usepackage{verbatim}
\newfont{\ffont}{cmr10}

\allowdisplaybreaks

\newcommand{\cH}{\mathcal{H}}

\newcommand{\cP}{\mathcal{P}}
\newcommand{\cF}{\mathcal{F}}
\newcommand{\cS}{\mathcal{S}}

\newcommand{\cO}{\mathcal{O}}
\newcommand{\cX}{\mathcal{X}}
\newcommand{\cU}{\mathcal{U}}
\newcommand{\cT}{\mathcal{T}}

\newcommand{\cM}{\mathcal{M}}

\newcommand{\al}{\alpha}

\newcommand{\om}{\omega}
\newcommand{\ga}{\gamma}
\newcommand{\del}{\delta}
\newcommand{\ka}{\kappa}
\newcommand{\Om}{\Omega}
\newcommand{\si}{\sigma}

\newcommand{\La}{\Lambda}
\newcommand{\Ups}{\Upsilon}
\newcommand{\Del}{\Delta}
\newcommand{\DelO}{\bar{\Delta}}
\newcommand{\R}{\mathbb{R}}

\newcommand{\N}{\mathbb{N}}

\newcommand{\E}{\mathbb{E}}

\newcommand{\bP}{\mathbb{P}}

\newcommand{\la}{\lambda}
\newcommand{\eps}{\varepsilon}
\newcommand{\ti}{\times}

\newcommand{\argmin}{\mathrm{argmin}}

\newcommand{\Fin}{\mathrm{Fin}}
\newcommand{\rank}{\mathrm{Rank}}
\newcommand{\Ch}{\mathrm{Ch}}

\newcommand{\xspace}{\hbox{\kern-2.5pt}}

\begin{document}

\theoremstyle{definition}
\newtheorem{definition}{Definition}[section]
\newtheorem{example}[definition]{Example}
\newtheorem{exercise}[definition]{Exercise}
\newtheorem{remark}[definition]{Remark}

\theoremstyle{plain}
\newtheorem{theorem}[definition]{Theorem}
\newtheorem{conjecture}[definition]{Conjecture}
\newtheorem{proposition}[definition]{Proposition}
\newtheorem{corollary}[definition]{Corollary}
\newtheorem{lemma}[definition]{Lemma}

\title[Dimensionality reduction with subgaussian matrices]{Dimensionality reduction with subgaussian matrices: a unified theory}
\author{Sjoerd Dirksen}

\address{Universit\"{a}t Bonn\\
Hausdorff Center for Mathematics\\
Endenicher Allee 60\\
53115 Bonn\\
Germany} \email{sjoerd.dirksen@hcm.uni-bonn.de}

\thanks{This research was supported by SFB grant 1060 of the Deutsche Forschungsgemeinschaft (DFG)}
\maketitle

\begin{abstract}
We present a theory for Euclidean dimensionality reduction with subgaussian matrices which unifies several restricted isometry property and Johnson-Lindenstrauss type results obtained earlier for specific data sets. In particular, we recover and, in several cases, improve results for sets of sparse and structured sparse vectors, low-rank matrices and tensors, and smooth manifolds. In addition, we establish a new Johnson-Lindenstrauss embedding for data sets taking the form of an infinite union of subspaces of a Hilbert space.
\end{abstract}

\section{Introduction}

The analysis of high-dimensional data leads to various computational issues which are gathered informally under the term `curse of dimensionality'. To circumvent such issues, several methods have been proposed to reduce the dimensionality of the data, i.e., to map the data set in a lower-dimensional space while approximately preserving certain relevant properties of the set. This is often possible as many high-dimensional data sets possess some additional structure which ensures that it has a low `intrinsic dimension' or `complexity'.\par
This paper concerns the \emph{random} dimensionality reduction method, which seeks to embed a data set using a random linear map. Undoubtedly the most famous known result in this direction is the classical embedding of Johnson and Lindenstrauss \cite{JoL84}. They showed that if $\Phi$ is the orthogonal projection onto an $m$-dimensional subspace of $\R^n$ which is chosen uniformly at random, then with high probability $\Phi$ preserves the pairwise distances in a given finite subset $\cP$ of $\R^n$ up to a multiplicative (or relative) error $\eps$ provided that $m\geq C\eps^{-2}\log|\cP|$, where $|\cP|$ is the cardinality of $\cP$. Simpler proofs of this result later appeared in \cite{DaG03,FrM88}. Due to these historic origins, many authors refer to random dimensionality reduction as the `random projection method'. However, it is well-known that one can replace the random projection matrix $\Phi$ by a computationally more attractive subgaussian matrix. These matrices perform equally well, as stated in the following modernized version of the Johnson-Lindenstrauss embedding \cite{InN07,Mat08}.
\begin{theorem}[Johnson-Lindenstrauss embedding]
\label{thm:JLclassicalIntro}
Let $\cP$ be a set of $|\cP|$ points in $\R^n$. Let $\tilde{\Phi}$ be an $m\ti n$ matrix with entries $\tilde{\Phi}_{ij}$ which are independent, mean-zero, unit variance and $\sqrt{\al}$-subgaussian. Set $\Phi=\tfrac{1}{\sqrt{m}}\tilde{\Phi}$. There exists an absolute constant $C>0$ such that for any given $0<\eps,\eta<1$ we have
\begin{equation}
\label{eqn:distPresIntro}
(1-\eps)\|x-y\|_2^2 \leq \|\Phi(x) - \Phi(y)\|_2^2 \leq (1+\eps)\|x-y\|_2^2 \qquad \mathrm{for \ all} \ x,y\in\cP
\end{equation}
with probability $1-\eta$, provided that
\begin{equation}
\label{eqn:condmIntro}
m\geq C \al^2\eps^{-2} \ \max\{\log |\cP|, \log(\eta^{-1})\}.
\end{equation}
\end{theorem}
If $\eps_{\cP,\Phi}$ denotes the smallest constant such that (\ref{eqn:distPresIntro}) holds, then we can formulate Theorem~\ref{thm:JLclassicalIntro} compactly by saying that $\bP(\eps_{\cP,\Phi}\geq \eps)\leq \eta$ if $m$ satisfies (\ref{eqn:condmIntro}). All results in this paper will be phrased in this manner. In general, the dependence on the number of points $|\cP|$ in (\ref{eqn:condmIntro}) cannot be improved \cite{Alo03} and the dependence on $\eps$ and $\eta$ is already optimal if $|\cP|=1$ \cite{JaW13}.\par
Random dimensionality reduction is attractive for several reasons. It is easy to implement and computationally inexpensive in comparison with other dimensionality reduction methods such as principal component analysis, see e.g.\ \cite{BiM01} for an empirical comparison on image and text data. Moreover, the method is \emph{non-adaptive} or \emph{oblivious} to the data set, meaning that the method does not require any prior knowledge of the data set as input. It is therefore particularly suitable as a preprocessing step. Methods incorporating random dimensionality reduction have been proposed for a wide range of tasks, such as approximate nearest-neighbor search \cite{InM98,InN07}, learning mixtures of Gaussians \cite{Das99,KMV12}, clustering \cite{BDL08,Sch00}, manifold learning \cite{BHW07,Ver13}, matched field processing \cite{Man13,MRS12} and least squares regression \cite{DMM11,MaM12}. Various other applications can be found in \cite{Vem04}. Sometimes these methods are coined a `compressive' version of the original method, e.g.\ compressive matched field processing \cite{MRS12}. Several of these applications rely on extensions of Theorem~\ref{thm:JLclassicalIntro} to an infinite, but structured data set $\cP$. In these results the factor $\log|\cP|$ in (\ref{eqn:condmIntro}) is replaced by a different quantity that represents the intrinsic dimension of the data set $\cP$. For instance, if $\cP$ is a $K$-dimensional subspace, then one can show that $\bP(\eps_{\cP,\Phi}\geq \eps)\leq \eta$ if $m\geq C \al^2\eps^{-2} \ \max\{K, \log(\eta^{-1})\}$. Such Johnson-Lindenstrauss embeddings for subspaces were introduced in \cite{Sar06} for use in numerical linear algebra. We also mention the embedding results for smooth manifolds \cite{BaW09,Cla08,EfW13} which are motivated by manifold learning. In a slightly different vein, some authors have investigated lower bounds on $m$ that guarantee that a subgaussian matrix preserves pairwise distances up to an \emph{additive} rather than a multiplicative error \cite{AHY13,InN07}.\par
In the signal processing literature some closely related results appear in the form of \emph{restricted isometry properties}. Recall that a map $\Phi:\R^n\rightarrow \R^m$ satisfies a restricted isometry property with constant $\del$ on a set $\cP$ if
$$(1-\del)\|x\|_2^2 \leq \|\Phi(x)\|_2^2 \leq (1+\del)\|x\|_2^2 \qquad \mathrm{for \ all} \ x\in\cP.$$
It is a well-known result from compressed sensing that a subgaussian matrix $\Phi$ satisfies the restricted isometry property on the set of $s$-sparse vectors in $\R^n$ with probability $1-\eta$ if $m\geq C \al^2\del^{-2} \ \max\{s\log(n/s), \log(\eta^{-1})\}$ \cite{BDD08,CaT06,Don06,MPT07,MPT08}. This property implies that with high probability one can recover any $s$-sparse signal in a stable, robust and algorithmically efficient manner from $m\sim s\log(n/s)$ subgaussian measurements, see e.g.\ \cite[Chapter 6]{FoR13} and the references therein. Inspired by these developments restricted isometry properties of subgaussian matrices have been established for various signal sets with structured sparsity \cite{BlD09,ElM09,GNE13}. Another example is the restricted isometry property of subgaussian matrices acting on low-rank matrices. This result is used as a substitute for the restricted isometry property on sparse vectors in the low-rank matrix recovery literature \cite{CaP11,RFP10}.\par
The purpose of this paper is to give a unified treatment of the aforementioned collection of restricted isometry and Johnson-Lindenstrauss type properties for subgaussian matrices, as well as their `additive error counterparts'. In Theorem~\ref{thm:KappaGeneral} we formulate a `master bound' from which one can deduce these properties for subgaussian maps (as in Definition~\ref{def:subgaussianMap}) acting on any given data set in a possibly infinite-dimensional Hilbert space. This result is an extension of earlier work in \cite{Gor88,KlM05,MPT07}, see the discussion after Theorem~\ref{thm:KappaGeneral} for details. We give a transparent proof using a new tail bound for suprema of empirical processes from \cite{Dir13}. The main focus of our work is to make Theorem~\ref{thm:KappaGeneral} an accessible tool for non-specialists, by demonstrating extensively how to apply it to extract results for concrete data structures. As it turns out, in all considered applications we recover the best known results in the literature and often find an improved lower bound on the target dimension $m$. On several occasions we also extend earlier results for Gaussian matrices to general subgaussian matrices. This class contains computationally more efficient matrices than Gaussian matrices, see \cite{Ach03} and Example~\ref{exa:Achlioptas}. Moreover, the extension to subgaussian matrices is of interest for certain signal processing applications, see \cite[Section 1.2]{FoR13} and \cite{RWH10} for examples.\par
To conclude, we give a brief overview of the considered applications. In Section~\ref{sec:covDim} we consider sets with low covering dimension. This class of sets includes finite unions of subspaces, with sets of sparse and cosparse vectors as particular examples, as well as low-rank matrices and tensors. In Section~\ref{sec:unionsubspaces} we consider sets forming an \emph{infinite} union of subspaces of a Hilbert space. A wide variety of models in signal processing can be expressed in this form. For instance, signals exhibiting structured sparsity, piecewise polynomials, certain finite rate of innovations models and overlapping echoes can be described in this fashion \cite{Blu11,BlD09,ElM09,LuD08}. The main result in this section, Theorem~\ref{thm:embedUnion}, establishes a new Johnson-Lindenstrauss type embedding for an infinite union of subspaces of a Hilbert space. The embedding dimension in this result depends on the maximal dimension of the subspaces and the complexity of the index set measured in terms of the Finsler distance, which is related to the largest principal angles between the subspaces. Our result significantly improves upon recent work in this direction \cite{Man13,MaR13}, see the discussion after Theorem~\ref{thm:embedUnion}. By combining with the main result of \cite{Blu11} we deduce that one can robustly reconstruct a signal in an infinite union of subspaces from a small number of subgaussian measurements using a generalized iterative hard thresholding method, see Remark~\ref{rem:Landweber}. Finally, in Section~\ref{sec:manifolds} we deduce three different dimensionality reduction results for smooth submanifolds of $\R^n$. We first deduce a guarantee under which lengths of curves in the manifold are preserved uniformly by a subgaussian map. Further on we give conditions under which pairwise ambient distances are preserved up to an additive error. Finally, we establish Johnson-Lindenstrauss embedding results for smooth manifolds. We first give an improvement of the embedding result for manifolds with low linearization dimension from \cite{AHY13}. In Theorems~\ref{thm:manifoldIota} and \ref{thm:manifoldReach} we present embeddings in the spirit of \cite{BaW09,Cla08,EfW13}. In particular, we extend the recent result of \cite{EfW13} from Gaussian to subgaussian matrices and achieve optimal scaling in the error parameter $\eps$.

\section{Preliminaries and notation}

Throughout the paper we use the following terminology. We use $(\Om,\cF,\bP)$ to denote a probability space and write $\E$ for the expected value. For a real-valued random variable $X$ we define its $\psi_2$ or \emph{subgaussian norm} by
$$\|X\|_{\psi_{2}} = \inf\{C>0 \ : \ \E\exp(|X|^2/C^2)\leq 2\}.$$
If $\|X\|_{\psi_{2}}<\infty$ then we call $X$ a \emph{subgaussian} random variable. In particular, any centered Gaussian random variable $g$ with variance $\si^2$ is subgaussian and $\|g\|_{\psi_{2}}\lesssim \si$. Also, if $|X|$ is bounded by $K$ then $X$ is subgaussian and $\|X\|_{\psi_2}\lesssim K$. We call a random vector $X:\Om\rightarrow \R^n$ subgaussian if
$$\sup_{\|x\|_2\leq 1}\|\langle X,x\rangle\|_{\psi_2}<\infty.$$
We say that $X$ is \emph{isotropic} if
$$\E\langle X,x\rangle^2 = \|x\|_2^2, \qquad \mathrm{for \ all} \ x \in \R^n.$$
If $T$ is a set, then $d:T\ti T\rightarrow \R_+$ is called a \emph{semi-metric} on $T$ if $d(x,y)=d(y,x)$ and $d(x,z) \leq d(x,y) + d(y,z)$ for all $x,y,z\in T$. If $S\subset T$, then we use
$$\Del_d(S) = \sup_{s,t \in S} d(s,t)$$
to denote its \emph{diameter}.\par
We conclude by fixing some notation. We use $\|\cdot\|_2$ to denote the Euclidean norm on $\R^n$ and let $d_2$ denote the associated Euclidean metric. If $H$ is a Hilbert space then $\langle\cdot,\cdot\rangle$ denotes the inner product on $H$, $\|\cdot\|_H$ the induced inner product and $d_H$ the induced metric on $H$. If $T\subset H$ we use $\Del_H(T)$ to denote its diameter. If $A:H_1\rightarrow H_2$ is a bounded linear operator between two Hilbert spaces then $\|A\|$ denotes its operator norm. If $S$ is a set then we let $|S|$ denote its cardinality. Given $0<\al<\infty$ we write $\log^{\al}x:=(\log(x))^{\al}$ and $\log_+(x) := \max\{\log(x),0\}$ for brevity. Finally, we write $A\lesssim B$ if $A\leq C B$ for some universal constant $C>0$ and write $A\simeq B$ if both $A\lesssim B$ and $A\gtrsim B$ hold.

\section{$\ga_2$-functional, Gaussian width and entropy}
\label{sec:gamma2}

In this section we discuss the $\ga_2$-functional of a semi-metric space $(T,d)$, which plays a central role in the formulation of Theorem~\ref{thm:KappaGeneral}. Intuitively, one should think of $\ga_2(T,d)$ as measuring the complexity of $(T,d)$.
\begin{definition}
\label{def:gamma2functional}
Let $(T,d)$ be a semi-metric space. A sequence $\cT=(T_n)_{n\geq 0}$ of $T$ is called \emph{admissible} if $|T_0|=1$ and $|T_n|\leq 2^{2^n}$. The \emph{$\ga_2$-functional} of $(T,d)$ is defined by
$$\ga_2(T,d) = \inf_{(\cT,\pi)}\sup_{t\in T}\sum_{n\geq 0} 2^{n/2}d(t,\pi_n(t)),$$
where the infimum is taken over all admissible sequences $\cT=(T_n)_{n\geq 0}$ in $T$ and all sequences $\pi=(\pi_n)_{n\geq 0}$ of maps $\pi_n:T\rightarrow T_n$.
\end{definition}
In the literature it is common to take $\pi_n(t):=\argmin_{s\in T_n}d(t,s)$ in Definition~\ref{def:gamma2functional}, i.e., to define the $\ga_2$-functional as
$$\ga_2(T,d) = \inf_{\cT}\sup_{t\in T}\sum_{n\geq 0} 2^{n/2}d(t,T_n),$$
where $d(t,T_n)=\inf_{s\in T_n} d(t,s)$. Our slightly relaxed definition will be convenient later on.\par
Let us recall the role of the $\ga_2$-functional in the theory of generic chaining. We recall two results from this theory. Suppose that $(X_t)_{t\in T}$ is a real-valued stochastic process, which has subgaussian increments with respect to a semi-metric $d$. That is, for all $s,t\in T$,
\begin{equation}
\label{eqn:subgProc}
\bP(|X_t - X_s|\geq ud(t,s)) \leq 2\exp(-u^2) \qquad (u\geq 0).
\end{equation}
Talagrand's generic chaining method \cite{Tal05} yields
$$\E\sup_{t\in T}|X_t| \lesssim \ga_2(T,d).$$
This bound is known to be sharp in the following interesting special case. Suppose that $(G_t)_{t\in T}$ is a centered Gaussian process and let $d_{\mathrm{can}}(s,t)=(\E|G_s-G_t|^2)^{1/2}$ be the induced canonical metric on $T$. Then, Talagrand's celebrated majorizing measures theorem \cite{Tal87,Tal01} states that
\begin{equation}
\label{eqn:majMeas}
\E\sup_{t\in T}|G_t| \simeq \ga_2(T,d_{\mathrm{can}}).
\end{equation}
Let $g=(g_1,\ldots,g_n)$ be a vector consisting of independent standard Gaussian variables and for any $x\in \R^n$ define $G_x = \langle g,x\rangle$. Then $(G_x)_{x\in T}$ is a centered Gaussian process for any given subset $T$ of $\R^n$. Note that the canonical metric $d_{\mathrm{can}}$ coincides with the usual Euclidean metric $d_2$ in this case. Hence, (\ref{eqn:majMeas}) translates into
\begin{equation}
\label{eqn:gaussWidth}
\ga_2(T,d_2) \simeq \E\sup_{x\in T} |\langle g,x\rangle|.
\end{equation}
The quantity on the right hand side is known as the \emph{Gaussian width} of the set $T$.\par
The $\ga_2$-functional can be estimated using covering numbers. For any given $u>0$ let $N(T,d,u)$ denote the covering number of $T$, i.e., the smallest number of balls of radius $u$ in $(T,d)$ needed to cover $T$. Then $\log N(T,d,u)$ is called the \emph{$u$-entropy} of $(T,d)$. Let
$$I_2(T,d) = \int_0^{\Del_d(T)} \log^{1/2} N(T,d,u) \ du$$
be the associated entropy integral. It is shown in \cite[Section 1.2]{Tal05} that
\begin{equation}
\label{eqn:gammaFunEstEntInt}
\ga_{2}(T,d) \lesssim I_2(T,d).
\end{equation}
and in particular $\ga_2(T,d)\lesssim \Del_d(T)\log^{1/2}|T|$ if $T$ is finite. The reverse estimate of (\ref{eqn:gammaFunEstEntInt}) fails \cite[Section 2.1]{Tal05}. However, if $T\subset \R^n$ then one can show that
$$I_2(T,d_2) \lesssim (\log n)\ga_2(T,d_2).$$
This worst-case bound is attained by natural objects, such as ellipsoids \cite[Section 2.2]{Tal05}. Even though (\ref{eqn:gammaFunEstEntInt}) is not sharp, it is a very important tool to estimate the $\ga_2$-functional in practical situations and we will use it several times below.\par
For our analysis we use the following tail bound for suprema of empirical processes from \cite[Theorem 5.5]{Dir13}. This result improves and extends two earlier results in the same direction of Klartag and Mendelson \cite{KlM05} and Mendelson, Pajor and Tomczak-Jaegermann \cite{MPT07}, see \cite{Dir13} for a detailed comparison.
\begin{theorem}
\label{thm:supAverages}
Fix a probability space $(\Om_i,\cF_i,\bP_i)$ for every $1\leq i\leq m$. For every $t\in T$ and $1\leq i\leq m$ let $X_{t,i}\in L^2(\Om_i)$. Define the process of averages
\begin{equation}
\label{eqn:AtDef}
A_t = \frac{1}{m} \sum_{i=1}^m X_{t,i}^2 - \E X_{t,i}^2.
\end{equation}
Consider the semi-metric
$$d_{\psi_2}(s,t) = \max_{1\leq i\leq m}\|X_{s,i}-X_{t,i}\|_{\psi_2} \qquad (s,t\in T)$$
and define the radius
$$\DelO_{\psi_2}(T) = \sup_{t\in T} \max_{1\leq i\leq m} \|X_{t,i}\|_{\psi_2}.$$
There exist constants $c,C>0$ such that for any $u\geq 1$,
\begin{align*}
\bP\Big(\sup_{t\in T}|A_t| & \geq C\Big(\frac{1}{m}\ga_2^2(T,d_{\psi_2}) + \frac{1}{\sqrt{m}}\DelO_{\psi_2}(T)\ga_2(T,d_{\psi_2})\Big) \\
& \qquad \qquad \qquad \qquad + c\Big(\sqrt{u}\frac{\DelO_{\psi_2}^2(T)}{\sqrt{m}} + u\frac{\DelO_{\psi_2}^2(T)}{m}\Big)\Big)\leq e^{-u}.
\end{align*}
\end{theorem}

\section{Master bound}
\label{sec:masterBound}

Throughout, let $H$ be a real Hilbert space. Let $\cP$ be a set of points in $H$. We are interested in reducing the dimensionality of $\cP$, meaning that we would like to construct a map $\Phi:H\rightarrow\R^m$ with the \emph{target dimension} $m$ as small as possible. We call the dimension of $H$ the \emph{original dimension}, which we think of as being very large or even infinite. All the results in this paper provide a lower bound on $m$ under which $\cP$ can be mapped into $\R^m$, while preserving certain properties of the set $\cP$. The following definition expresses that $\Phi$ approximately preserves the \emph{size} of the original vectors.
\begin{definition}
Let $\cP$ be a set in $H$ and let $\Phi:H\rightarrow \R^m$. The \emph{restricted isometry constant} $\del_{\cP,\Phi}$ of $\Phi$ on $\cP$ is the least possible constant $0\leq \del\leq\infty$ such that
$$(1-\del)\|x\|_H^2 \leq \|\Phi(x)\|_2^2 \leq (1+\del)\|x\|_{H}^2.$$
\end{definition}
It is common parlance to loosely say that a map $\Phi$ satisfies \emph{the restricted isometry property} on $\cP$ if $\del_{\cP,\Phi}<\del_*$, where $\del_*<1$ is some small value.\par
The following definition expresses that $\Phi$ preserves \emph{pairwise distances} between elements of $\cP$. We distinguish between a \emph{multiplicative} and an \emph{additive} error.
\begin{definition}
\label{def:pairwiseDist}
Let $\cP$ be a set in $H$ and let $\Phi:H\rightarrow \R^m$. We say that $\Phi$ preserves distances on $\cP$ with \emph{multiplicative} error $0<\eps<1$ if
\begin{equation}
\label{eqn:defMulError}
(1-\eps)\|x-y\|_{H}^2 \leq \|\Phi(x)-\Phi(y)\|_2^2 \leq (1+\eps)\|x-y\|_{H}^2 \qquad \mathrm{for \ all} \ x,y\in\cP.
\end{equation}
The least possible constant $\eps_{\cP,\Phi}$ for which this holds is called the \emph{multiplicative precision} of $\Phi$. We say that $\Phi$ preserves distances on $\cP$ with \emph{additive} error $0\leq \zeta<1$ if
\begin{equation}
\label{eqn:defAddError}
\|x-y\|_{H}^2 - \zeta \leq \|\Phi(x)-\Phi(y)\|_2^2 \leq \|x-y\|_{H}^2 + \zeta \qquad \mathrm{for \ all} \ x,y\in\cP.
\end{equation}
The least possible constant $\zeta_{\cP,\Phi}$ for which this holds is called the \emph{additive precision} of $\Phi$.
\end{definition}
The restricted isometry constant and the multiplicative error in Definition~\ref{def:pairwiseDist} are closely related. If
\begin{equation}
\label{eqn:defChords}
\cP_c=\{x-y \ : \ x,y\in \cP\}
\end{equation}
denotes the set of chords associated with $\cP$, then $\eps_{\cP,\Phi} = \del_{\cP_c,\Phi}$.\par
Maps $\Phi$ that preserve pairwise distances up to a multiplicative error $\eps$ are the most interesting for applications as they often preserve additional properties of $\cP$. For instance, in applications it is often used that if $\Phi$ is in addition linear and $-\cP=\cP$ then inner products are preserved up to additive error, i.e.,
$$|\langle\Phi(x),\Phi(y)\rangle - \langle x,y\rangle|\leq \eps$$
for all unit vectors $x,y \in \cP$. This can be readily shown using a polarization argument. More can be said if $\cP$ is a manifold, see Section~\ref{sec:manifolds} below.
\begin{remark}
\label{rem:noSquaresMult}
In parts of the literature, it is customary to say that a map $\Phi$ approximately preserves distances on $\cP$ with multiplicative error $0<\hat{\eps}<1$ if
\begin{equation}
\label{eqn:defMulErrorNoSq}
(1-\hat{\eps})\|x-y\|_{H} \leq \|\Phi(x)-\Phi(y)\|_2 \leq (1+\hat{\eps})\|x-y\|_{H} \qquad \mathrm{for \ all} \ x,y\in\cP.
\end{equation}
That is, one leaves out the squares in (\ref{eqn:defMulError}). Let $\hat{\eps}_{\cP,\Phi}$ be the smallest possible $\hat{\eps}$ in (\ref{eqn:defMulErrorNoSq}). One readily checks that $\hat{\eps}_{\cP,\Phi}\leq \hat{\eps}$ if $\eps_{\cP,\Phi}\leq 2\hat{\eps}-\hat{\eps}^2$.
\end{remark}
Below we will derive various dimensionality reduction results for the following class of random maps.
\begin{definition}[Subgaussian map]
\label{def:subgaussianMap}
Let $\cS$ be any set of points in $H$. For every $1\leq i\leq m$ let $(\Om_i,\cF_i,\bP_i)$ be a probability space. Let $\Om$ be the corresponding product probability space. We call $\Phi:\Om\ti H\rightarrow \R^m$ a linear, isotropic, subgaussian map, or briefly a \emph{subgaussian map on $\cS$} if the following conditions hold.
\begin{enumerate}
\item (Linearity) For any $\om \in \Om$ the map $\Phi(\om):H\rightarrow \R^m$ is linear;
\item (Independence) For all $x\in \cS$ and $1\leq i\leq m$, $[\Phi(x)]_i \in L^2(\Om_i)$;
\item (Isotropy) For any $x\in \cS$ we have $\E\|\Phi(x)\|_2^2 = \|x\|_{H}^2$;
\item (Subgaussianity) There is an $\al\geq 1$ such that for all $x,y \in \cS\cup\{0\}$,
$$\max_{1\leq i\leq m} \|[\Phi(x) - \Phi(y)]_i\|_{\psi_2} \leq \sqrt{\frac{\al}{m}} \|x-y\|_{H}.$$
\end{enumerate}
\end{definition}
Note that the condition $\al\geq 1$ is forced by the assumption that $\Phi$ is isotropic.
\begin{example}[Subgaussian matrices]
Suppose that $H=\R^n$ for some $n\in \N$, equipped with the Euclidean norm. Let $\tilde{\Phi}$ be an $m\ti n$ random matrix, whose rows $\tilde{\Phi}_1,\ldots,\tilde{\Phi}_m$ are independent, mean-zero, isotropic, subgaussian random vectors in $\R^n$. By setting $\Phi=\frac{1}{\sqrt{m}}\tilde{\Phi}$ we obtain a subgaussian map on $\R^n$.
\end{example}
\begin{example}[Database-friendly maps \cite{Ach03}]
\label{exa:Achlioptas}
As a particular instance of the previous example, we can let $\tilde{\Phi}$ be any random matrix filled with independent, mean-zero, unit variance, subgaussian (in particular, bounded) entries $\Phi_{ij}$. In \cite{Ach03}, Achlioptas proposed to take independent random variables satisfying
$$\bP(\Phi_{ij} = -\sqrt{3}) = \frac{1}{6}, \qquad \bP(\Phi_{ij} = 0) = \frac{2}{3}, \qquad \bP(\Phi_{ij} = \sqrt{3}) = \frac{1}{6}.$$
Due to the (expected) large number of zeroes occurring in $\Phi$, this map requires less storage space (or, as it is phrased in \cite{Ach03}, it is `database-friendly') and allows for faster matrix-vector multiplication than a densely populated matrix \cite[Section 7]{Ach03}. More generally, for any $q\geq 1$ one can take
$$\bP(\Phi_{ij} = -\sqrt{q}) = \frac{1}{2q}, \qquad \bP(\Phi_{ij} = 0) = \frac{q-1}{q}, \qquad \bP(\Phi_{ij} = \sqrt{q}) = \frac{1}{2q}.$$
One can readily compute that the subgaussian parameter $\al$ in part (d) of Definition~\ref{def:subgaussianMap} is bounded by $q$ in this case.
\end{example}
\begin{example}[Infinite-dimensional Hilbert spaces]
Suppose that $H$ is a separable Hilbert space. Let $(x_j)_{j\geq 1}$ be any orthonormal basis of $H$ and, for $1\leq i\leq m$, let $(g_j^{(i)})_{j\geq 1}$ be independent sequences of i.i.d.\ standard Gaussian random variables. For $1\leq i\leq m$ we define  $\tilde{\Phi}_i:H\rightarrow L^2(\Om)$ by
$$\tilde{\Phi}_i x = \sum_{j=1}^{\infty} g_j^{(i)}\langle x,x_j\rangle.$$
The map $\Phi:\Om\ti H\rightarrow \R^m$ defined by $\Phi x=\frac{1}{\sqrt{m}}(\tilde{\Phi}_1 x,\ldots,\tilde{\Phi}_m x)$ is subgaussian on $H$.
\end{example}
To give a unified presentation of dimensionality reduction results for operators which are size and pairwise distance preserving, we define for a given set $\cS$ in $H$ the constant $0\leq\ka_{\cS,\Phi}\leq\infty$ by
$$\ka_{\cS,\Phi} = \sup_{x\in \cS} \Big|\|\Phi(x)\|_2^2 - \|x\|_H^2\Big|.$$
The restricted isometry constant of $\Phi$ on $\cP$ is exactly $\ka_{\cP_{nv},\Phi}$, where
\begin{equation}
\label{eqn:RIconstKappa}
\cP_{nv}=\Big\{\frac{x}{\|x\|_H} \ : \ x\in \cP\Big\}
\end{equation}
is the set of normalized vectors in $\cP$. For a pairwise distance preserving operator $\Phi$ the multiplicative error on $\cP$ is equal to $\ka_{\cP_{nc},\Phi}$, where
$$\cP_{nc} = \Big\{\frac{x-y}{\|x-y\|_H} \ : \ x,y\in\cP\Big\}$$
is the set of normalized chords corresponding to $\cP$. The additive error on $\cP$ is equal to $\ka_{\cP_c,\Phi}$, with $\cP_c$ as in (\ref{eqn:defChords}).\par
For our treatment in Section~\ref{sec:unionsubspaces} it will be convenient to consider the situation where $\cS$ is described by a parameter set $\Xi$. We will say that $\xi:\Xi\rightarrow \cS$ is a \emph{parametrization of $\cS$} if $\xi$ is a surjective map. To any parametrization $\xi$ we associate a semi-metric $d_{\xi}$ on $\Xi$ defined by
\begin{equation}
\label{eqn:dxiDef}
d_{\xi}(x,y):=\|\xi(x)-\xi(y)\|_{H} \qquad (x,y\in \Xi).
\end{equation}
We can now state a `master bound'. Every dimensionality reduction result stated below is a corollary of this theorem.
\begin{theorem}
\label{thm:KappaGeneral}
Let $\cS$ be a set of points in $H$ with radius $\DelO_H(\cS)=\sup_{y\in \cS} \|y\|_H$ and let $\xi:\Xi\rightarrow \cS$ be a parametrization of $\cS$. Let $\Phi:\Om\ti H\rightarrow \R^m$ be a subgaussian map on $\cS$. There is a constant $C>0$ such that for any $0<\ka,\eta<1$ we have $\bP(\ka_{\cS,\Phi}\geq \ka)\leq\eta$ provided that
\begin{equation}
\label{eqn:KappageneralMeasPar}
m \geq C\al^2\ka^{-2}\DelO_H^2(\cS) \ \max\{\ga_2^2(\Xi,d_{\xi}), \DelO_H^2(\cS)\log(\eta^{-1})\}.
\end{equation}
\end{theorem}
\begin{proof}
For any $x\in \Xi$ we write $\Phi_i(x):=\sqrt{m}[\Phi(\xi(x))]_i$. By isotropy of $\Phi$,
$$\|\Phi(x)\|_2^2 - \|x\|_{H}^2 = \|\Phi(x)\|_2^2 - \E\|\Phi(x)\|_2^2 = \frac{1}{m}\sum_{i=1}^m \Phi_i(x)^2 - \E\Phi_i(x)^2.$$
We can now set $T=\Xi$ and $X_{t,i} = \Phi_i(x)$ in Theorem~\ref{thm:supAverages} to obtain for any $u\geq 1$
$$\bP\Big(\ka_{\cS,\Phi}\geq C\Big(\frac{\ga_2^2(\Xi,d_{\psi_2})}{m} + \frac{\ga_2(\Xi,d_{\psi_2})\DelO_{\psi_2}(\Xi)}{\sqrt{m}} + \sqrt{u}\frac{\DelO_{\psi_2}^2(\Xi)}{\sqrt{m}} + u\frac{\DelO_{\psi_2}^2(\Xi)}{m}\Big)\Big) \leq e^{-u}.$$
Since $\Phi$ is subgaussian we have for any $x,y\in \Xi$,
\begin{align*}
d_{\psi_2}(x,y) & = \max_{1\leq i\leq m} \|\Phi_i(x) - \Phi_i(y)\|_{\psi_2} \\
& = \sqrt{m}\max_{1\leq i\leq m} \|[\Phi(\xi(x))-\Phi(\xi(y))]_i\|_{\psi_2} \leq \sqrt{\al} \|\xi(x)-\xi(y)\|_{H} = \sqrt{\al} d_{\xi}(x,y)
\end{align*}
and similarly, for any $x\in \Xi$,
$$\max_{1\leq i\leq m} \|\Phi_i(\xi(x))\|_{\psi_2} \leq \sqrt{\al} \|\xi(x)\|_{H}\leq \sqrt{\al}\DelO_H(\cS).$$
In particular,
$$\ga_2(\Xi,d_{\psi_2}) \leq \sqrt{\al} \ga_2(\Xi,d_{\xi}), \qquad \DelO_{\psi_2}(\Xi) \leq \sqrt{\al}\DelO_H(\cS).$$
We conclude that $\bP(\ka_{\cP,\Phi}\geq \ka)\leq \eta$ if (\ref{eqn:KappageneralMeasPar}) holds.
\end{proof}
In the special case that $\cS$ is a subset of the unit sphere of $H$ and $\xi$ is the trivial parametrization, Theorem~\ref{thm:KappaGeneral} corresponds to a result of Mendelson, Pajor and Tomczak-Jaegermann \cite[Corollary 2.7]{MPT07}, in which the $\ga_2$-functional in (\ref{eqn:KappageneralMeasPar}) is replaced by the Gaussian width of $\cS$ (this is equivalent by (\ref{eqn:gaussWidth}) in this case). They refined important earlier work of Klartag and Mendelson \cite{KlM05}, who proved the same result with a suboptimal dependence in $\eta$. The result in \cite{MPT07} was obtained much earlier for Gaussian matrices by Gordon \cite{Gor88}. The proof given in \cite{MPT07} makes specific use of the assumption that $\cS$ is contained in the unit sphere and cannot be easily modified to cover the general case considered here.
\begin{remark}[Anisotropy]
One can relax the assumption that the subgaussian map $\Phi$ is isotropic on the set $\cS$. Suppose that $\Phi$ satisfies (a), (b) and (d) in Definition~\ref{def:subgaussianMap}. Set $\Psi=\E\Phi^*\Phi$, then
$$\E\|\Phi x\|_2^2 = x^*\E(\Phi^*\Phi)x = x^*\Psi x = \|\Psi^{1/2}x\|_H^2.$$
The proof of Theorem~\ref{thm:KappaGeneral} shows that for any $0<\ka,\eta<1$ we have
$$-\ka + \|\Psi^{1/2}x\|_H^2 \leq \|\Phi x\|_2^2 \leq \|\Psi^{1/2}x\|_H^2 + \ka, \qquad \mathrm{for \ all} \ x\in \cS$$
with probability at least $1-\eta$ provided that (\ref{eqn:KappageneralMeasPar}) holds.
\end{remark}
The following statements are immediate from Theorem~\ref{thm:KappaGeneral} by setting $\cS=\cP_{nv}$, $\cS=\cP_{nc}$ and $\cS=\cP_c$, respectively, and taking the trivial parametrization $\xi(x)=x$.
\begin{corollary}
\label{cor:RIPepsIsom}
Let $\cP$ be a set of points and let $\Phi$ be as in Theorem~\ref{thm:KappaGeneral}. For any $0<\del,\eta<1$ we have $\bP(\del_{\cP,\Phi}\geq \del)\leq\eta$ provided that
\begin{equation*}
m \geq C\al^2\del^{-2} \ \max\{\ga_2^2(\cP_{nv},d_{H}), \log(\eta^{-1})\}.
\end{equation*}
Moreover, for any $0<\eps,\eta<1$ we have $\bP(\eps_{\cP,\Phi}\geq \eps)\leq\eta$ whenever
\begin{equation*}
m \geq C\al^2\eps^{-2} \ \max\{\ga_2^2(\cP_{nc},d_{H}), \log(\eta^{-1})\}.
\end{equation*}
Finally, for any $0<\zeta,\eta<1$ we have $\bP(\zeta_{\cP,\Phi}\geq \zeta)\leq\eta$ if
$$m\geq C\al^2\zeta^{-2}\Del_{H}^2(\cP) \ \max\{\ga_2^2(\cP,d_{H}), \Del_{H}^2(\cP)\log(\eta^{-1})\}.$$
\end{corollary}
As a first illustration of Corollary~\ref{cor:RIPepsIsom}, note that it implies an extension of Theorem~\ref{thm:JLclassicalIntro} to general subgaussian maps as
$$\ga_2^2(\cP_{nc}) \lesssim \log|\cP_{nc}| \leq 2\log|\cP|.$$
Since the dependence of $m$ on $\eps$ and $\eta$ in (\ref{eqn:condmIntro}) is optimal, we see that in general one cannot expect a better dependence of $m$ on $\ka$ and $\eta$ in (\ref{eqn:KappageneralMeasPar}).\par
In the remainder of this paper we derive dimensionality reduction results for concrete data structures from Theorem~\ref{thm:KappaGeneral}. The technical work is to derive a good estimate for the complexity parameter $\ga_2^2(\Xi,d_{\xi})$ appearing in (\ref{eqn:KappageneralMeasPar}).

\section{Sets with low covering dimension}
\label{sec:covDim}

In this section we consider dimensionality reduction for sets with low covering dimension, in particular finite unions of subspaces.
\begin{definition}
We say that a metric space $(\cX,d)$ has \emph{covering dimension $K>0$ with parameter $c>0$ and base covering $N_0>0$} if, for all $0<u\leq 1$,
$$N(\cX,d,u\Del_{d}(\cX)) \leq N_0\Big(\frac{c}{u}\Big)^K.$$
\end{definition}
Often $c$ and $N_0$ are some small universal constants. In this situation we will loosely say that $(\cX,d)$ has covering dimension $K$.
\begin{example}[Unit ball of a finite-dimensional space]
\label{exa:covDimSubspace}
A well-known example is the unit ball $B_{\cX}$ of a $K$-dimensional normed space $\cX$. Using a standard volumetric argument (see e.g.\ \cite[Proposition C.3]{FoR13}) one shows that for any $0<u\leq 1$,
$$N(B_{\cX},d_{\cX},u) \leq \Big(1+\frac{2}{u}\Big)^K \leq \Big(\frac{3}{u}\Big)^K.$$
\end{example}
\begin{example}[Doubling dimension]
If $(\cX,d)$ is a metric space, then the \emph{doubling constant} $\la_{\cX}$ of $\cX$ is the smallest integer $\la$ such that for any $x\in \cX$ and $u>0$, the ball $B(x,u)$ can be covered by at most $\la$ balls of radius $u/2$. One can show that \cite{InN07} for all $0<u\leq 1$,
$$N(\cX,d,u\Del_d(X)) \leq \Big(\frac{2}{u}\Big)^{\log_2 \la_X}.$$
That is, $(\cX,d)$ has covering dimension $\log_2\la_X$. The latter number is also known as the \emph{doubling dimension} of $(X,d)$. The notion of doubling dimension was considered in the context of dimensionality reduction in, for example, \cite{AHY13} and \cite{InN07}.
\end{example}
We now formulate a dimensionality reduction result for sets with low covering dimension. In the proof we use that for $c,u_*>0$,
\begin{equation}
\label{eqn:intEstimate}
\int_0^{u_*} \log^{1/2}(c/u) \ du \leq u_*\log^{1/2}\Big(\frac{ec}{u_*}\Big).
\end{equation}
A short proof of this estimate can be found in \cite[Lemma C.9]{FoR13}. The second statement in the following result was obtained by a different method in \cite[Theorem 3]{Blu11} (note that this result was erroneously stated in \cite{BlD09}), but with a suboptimal dependence on $\del$.
\begin{corollary}
\label{cor:covDim}
Let $S_1,\ldots,S_k$ be subsets of a Hilbert space $H$ and let $S=\cup_{i=1}^k S_i$. Set
$$S_{i,nv} = \{x/\|x\|_2 \ : \ x\in S_i\}.$$
Suppose that $S_{i,nv}$ has covering dimension $K_i$ with parameter $c_i$ and base covering $N_{0,i}$ with respect to $d_H$. Set $K=\max_i K_i$, $c=\max_i c_i$ and $N_0=\max_i N_{0,i}$. Let $\Phi:\Om\ti H\rightarrow \R^m$ be a subgaussian map on $S_{nv}$. Then, for any $0<\del,\eta<1$ we have $\bP(\del_{S,\Phi}\geq \del)\leq\eta$ provided that
\begin{equation*}
m \geq C\al^2\del^{-2} \ \max\{\log k + \log N_0+ K\log(c), \log(\eta^{-1})\}.
\end{equation*}
In particular, if each $S_i$ is a $K_i$-dimensional subspace of $\R^n$, then $\bP(\del_{S,\Phi}\geq \del)\leq\eta$ if
\begin{equation*}
m \geq C\al^2\del^{-2} \ \max\{\log k + K, \log(\eta^{-1})\}.
\end{equation*}
\end{corollary}
From Corollary~\ref{cor:covDim} one can readily deduce that $\bP(\eps_{S,\Phi}\geq \eps)\leq\eta$ if
$$m \geq C\al^2\del^{-2} \ \max\{\log k + \log N_0+ K\log(c), \log(\eta^{-1})\},$$
see the proof of Theorem~\ref{thm:embedUnion} below.
\begin{proof}
We use (\ref{eqn:gammaFunEstEntInt}) to estimate
$$\ga_{2}(S_{nv},d_H) \lesssim \int_0^{1} \log^{1/2} N(S_{nv},d_H,u) \ du.$$
Clearly, if $N_i$ is an $u$-net for $S_{i,nv}$, then $\cup_{i=1}^k N_i$ is an $u$-net for $S_{nv}$. Therefore, using our assumption on the $S_{i,nv}$, we obtain
\begin{equation*}
N(S_{nv},d_H,u) \leq \sum_{i=1}^k N(S_{i,nv},d_H,u) \leq \sum_{i=1}^k \Big(\frac{c_i}{u}\Big)^{K_i} \leq k\Big(\frac{c}{u}\Big)^{K}.
\end{equation*}
Using (\ref{eqn:intEstimate}) we arrive at
\begin{align*}
\ga_{2}(S_{nv},d_H) & \lesssim \int_0^{1} \log^{1/2}(k(c/u)^{K}) \ du \\
& \leq \log^{1/2}(k) + K^{1/2}\int_0^{1} \log^{1/2}(c/u) \ du \lesssim \log^{1/2}(k) + K^{1/2}\log^{1/2}(c).
\end{align*}
The first part of the result now follows from the first statement in Corollary~\ref{cor:RIPepsIsom} and the second part follows by the observation in Example~\ref{exa:covDimSubspace}.
\end{proof}
To illustrate Corollary~\ref{cor:covDim}, we consider four examples.
\begin{example}[Sparse vectors: the `usual' RIP]
\label{exa:usualRIP}
We derive the restricted isometry property on $s$-sparse vectors for subgaussian maps $\Phi:\Om\ti\R^n\rightarrow\R^m$, a classical result from compressed sensing \cite{BDD08,CaT06,Don06,MPT07,MPT08}. For $x\in \R^n$ we set $$\|x\|_0 = |\{1\leq i\leq n \ : \ x_i\neq 0\}|.$$
A vector is called $s$-sparse if $\|x\|_0\leq s$. Let
$$D_{s,n} = \{x \in \R^n \ : \ \|x\|_0\leq s\}$$
be the set of $s$-sparse vectors. The restricted isometry constant $\del_s$ of $\Phi$ is defined as the smallest constant $\del$ such that
$$(1-\del)\|x\|_2^2 \leq \|\Phi x\|_2^2 \leq (1+\del)\|x\|_2^2 \qquad \mathrm{for \ all} \ x\in D_{s,n}.$$
In our notation, $\del_s = \del_{D_{s,n},\Phi}$. Note that we can write
$$D_{s,n} = \bigcup_{I\subset\{1,\ldots,n\}, |I|=s} S_I,$$
where $S_I$ is the $s$-dimensional subspace
$$S_I = \{x \in \R^n \ : \ x_i = 0 \ \mathrm{if} \ i\in I^c\}.$$
Since the number of $s$-element subsets of $\{1,\ldots,n\}$ is
$$\binom{n}{s} \leq \Big(\frac{en}{s}\Big)^s,$$
the second part of Corollary~\ref{cor:covDim} implies that $\bP(\del_s\geq \del)\leq\eta$ provided that
\begin{equation*}
m \geq C\al^2\del^{-2} \ \max\{s\log(en/s), \log(\eta^{-1})\}.
\end{equation*}
The scaling of this lower bound in $n$ and $s$ is optimal, see \cite[Corollary 10.8]{FoR13}.
\end{example}
\begin{example}[Cosparse vectors: the $\Psi$-RIP]
In many signal processing applications, signals of interest are not sparse themselves in the standard basis, but can rather be represented as a sparse vector. Let $\Upsilon:\R^n\rightarrow\R^p$ be a linear operator, which is usually called the `analysis operator' in the literature. We are interested in elements $x\in \R^n$ such that $\Ups x$ is sparse. It has become customary to count the number of zero components of $\Ups x$, rather than the number of nonzero ones. Accordingly, a vector $x\in \R^n$ is called \emph{$l$-cosparse with respect to $\Ups$} if there is a set $\La\subset \{1,\ldots,p\}$ with $|\La|=l$ such that $\Ups_{\La} x=0$, where $\Ups_{\La}:\R^{n}\rightarrow \R^p$ is the operator obtained by setting the rows of $\Ups$ indexed by $\Lambda^c$ equal to zero. Let $N_{\La}$ be the null space of $\Ups_{\La}$, then we can write the set of $l$-cosparse vectors as
$$C_{\Ups,l,p}=\bigcup_{|\La|=l} N_{\La}.$$
The \emph{$\Ups$-restricted isometry constant $\del_l$} of $\Phi$ is defined as the smallest possible $\del>0$ such that
$$(1-\del)\|x\|_2^2 \leq \|\Phi x\|_2^2 \leq (1+\del)\|x\|_2^2 \qquad \mathrm{for \ all} \ x\in C_{\Ups,l,p}.$$
In our notation, $\del_l = \del_{C_{\Ups,l,p},\Phi}$. Observe that $\dim(N_{\La})\leq n-l$ and the number of $l$-element subsets of $\{1,\ldots,p\}$ is
$$\binom{p}{l} \leq \Big(\frac{ep}{l}\Big)^l.$$
The second part of Corollary~\ref{cor:covDim} now implies that $\bP(\del_l\geq \del)\leq\eta$ if
\begin{equation*}
m \geq C\al^2\del^{-2} \ \max\{l\log(ep/l) + (n-l), \log(\eta^{-1})\}.
\end{equation*}
This result improves upon the RIP-result in \cite[Theorem 3.8]{GNE13}. In fact, as is heuristically explained in \cite[Section 6.1]{NDE13}, one cannot expect a better lower bound for $m$. An upper bound on $\del_l$ leads to performance guarantees for greedy-like recovery algorithms for cosparse vectors from subgaussian measurements, see \cite{GNE13} for some results in this direction.
\end{example}
\begin{example}[Matrix RIP]
Another direct consequence of Corollary~\ref{cor:covDim} is a new proof of the restricted isometry property for subgaussian matrix maps. This property plays the same role in low-rank matrix recovery as the `usual' restricted isometry property discussed in Example~\ref{exa:usualRIP} plays in compressed sensing, see e.g.\ \cite{CaP11,RFP10} for further information.\par
We use the following notation. Given two matrices $X,Y \in \R^{n_1\ti n_2}$ we consider the Frobenius inner product
$$\langle X,Y\rangle = \sum_{i=1}^{n_1}\sum_{j=1}^{n_2} X_{ij}Y_{ij}.$$
Let $\|X\|_F = \langle X,X\rangle^{1/2}$ be the corresponding norm and $d_F(X,Y) = \|X-Y\|_F$ be the induced metric. Also, we use $\rank(X)$ to denote the rank of $X$. For $1\leq r\leq \min\{n_1,n_2\}$ we define the restricted isometry constant $\del_r$ of a map $\Phi:\R^{n_1\ti n_2}\rightarrow \R^m$ as the smallest constant $\del>0$ such that
$$(1-\del)\|X\|_F^2 \leq \|\Phi X\|_F^2 \leq (1+\del)\|X\|_F^2 \qquad \mathrm{for \ all} \ X\in \R^{n_1\ti n_2} \ \mathrm{with} \ \rank(X)\leq r.$$
If we set
$$D_{r} = \{X \in \R^{n_1\ti n_2} \ : \ \|X\|_F = 1, \rank(X)\leq r\},$$
then $\del_{r}=\del_{D_{r},\Phi}$ in our notation. The covering number estimate \cite[Lemma 3.1]{CaP11}
$$N(D_{r},d_F,u) \leq (9/u)^{r(n_1+n_2+1)} \qquad (0<u\leq 1),$$
shows that $D_{r}$ has covering dimension $r(n_1+n_2+1)$ in $(\R^{n_1\ti n_2},d_F)$. The first part of Corollary~\ref{cor:covDim} implies for any subgaussian map $\Phi$ that $\bP(\del_{r}\geq \del)\leq \eta$ if
$$m \geq C\al^2\del^{-2} \ \max\{r(n_1+n_2+1), \log(\eta^{-1})\}.$$
This result was obtained in a different way in \cite[Theorem 2.3]{CaP11}, see also \cite[Theorem 4.2]{RFP10} for a slightly worse result.
\end{example}
\begin{example}[Tensor RIP]
The previous example can be extended to higher order tensors. Let $d\geq 2$ and set $n=(n_1,\ldots,n_d) \in \N^d$. Given two tensors $X,Y \in \R^{n_1\ti\cdots\ti n_d}$ we consider their Frobenius inner product
$$\langle X,Y\rangle = \sum_{i_1=1}^{n_1}\cdots\sum_{i_d=1}^{n_d} X(i_1,\ldots,i_d)Y(i_1,\ldots,i_d).$$
Let $\|X\|_F = \langle X,X\rangle^{1/2}$ be the corresponding norm and $d_F(X,Y) = \|X-Y\|_F$ be the induced metric. Let $\rank(X)$ denote the rank of $X$ associated with its HOSVD decomposition, see \cite{RSS13} for more information.\par
Given $r=(r_1,\ldots,r_d)$, $0\leq r_i\leq n_i$, we define the restricted isometry constant $\del_r$ of a map $\Phi:\R^{n_1\ti\cdots\ti n_d}\rightarrow \R^m$ as the smallest constant $\del>0$ such that
$$(1-\del)\|X\|_F^2 \leq \|\Phi X\|_F^2 \leq (1+\del)\|X\|_F^2 \ \ \ \mathrm{for \ all} \ X\in \R^{n_1\ti\cdots\ti n_d} \ \mathrm{with} \ \rank(X)\leq r.$$
If we set
$$D_{r} = \{X \in \R^{n_1\ti\cdots\ti n_d} \ : \ \|X\|_F = 1, \rank(X)\leq r\},$$
then $\del_{r}=\del_{D_{r},\Phi}$ in our notation. It is shown in \cite{RSS13} that for any $0<u\leq 1$
$$N(D_{r},d_F,u) \leq (3(d+1)/u)^{r_1\cdots r_d + \sum_{i=1}^d n_ir_i}.$$
In other words, $D_{r}$ has covering dimension $r_1\cdots r_d + \sum_{i=1}^d n_ir_i$ with parameter $3(d+1)$. Corollary~\ref{cor:covDim} implies that for any subgaussian map $\Phi$ and any $0<\del,\eta<1$ we have $\bP(\del_{r}\geq \del)\leq \eta$, provided that
\begin{equation*}
m\geq C \al^2\del^{-2} \ \max\Big\{\Big(r_1\cdots r_d + \sum_{i=1}^d n_ir_i\Big)\log(d), \log(\eta^{-1})\Big\}.
\end{equation*}
This result was obtained originally for a more restricted class of subgaussian maps in \cite{RSS13}.
\end{example}
The results presented in the four examples above can be derived in a different, more elementary fashion using the $\eps$-net technique, see \cite{BDD08,MPT08}, \cite{GNE13}, \cite{CaP11,RFP10}, and \cite{RSS13}, respectively. In fact, this is already true for the statement in Corollary~\ref{cor:covDim}. The results in the following two sections cannot be achieved using the $\eps$-net technique, however, and therefore generic chaining methods, which are at the basis of Theorem~\ref{thm:supAverages}, become necessary to achieve the best results.

\section{Infinite union of subspaces}
\label{sec:unionsubspaces}

Many sets of signals relevant to signal processing can be expressed as a possibly infinite union of finite-dimensional subspaces of a Hilbert space. For example, signals with various forms of structured sparsity (e.g.\ sparse, cosparse, block sparse, simultaneously sparse data), piecewise polynomials, certain finite rate of innovations models and overlapping echoes can be described in this fashion \cite{Blu11,BlD09,ElM09,LuD08}. In this section we prove a Johnson-Lindenstrauss embedding result for an infinite union of subspaces.\par
We consider the following setup. Let $H$ be a Hilbert space and let $B_H$ denote its unit ball. Let $\Theta$ be a parameter set and suppose that for every $\theta\in\Theta$ we are given a finite-dimensional subspace $S_{\theta}$ of $H$. We use $P_{\theta}$ to denote the orthogonal projection onto $S_{\theta}$. It will be natural to consider the Finsler metric on $\Theta$, which is defined by
$$d_{\Fin}(\theta,\theta') := \|P_{\theta} - P_{\theta'}\| \qquad (\theta,\theta' \in \Theta).$$
If two subspaces $S_{\theta}, S_{\theta'}$ have the same dimension, then the Finsler distance satisfies
$$d_{\Fin}(\theta,\theta') = \sin(\gamma(\theta,\theta')),$$
where $\gamma(\theta,\theta')$ is the largest canonical angle (or largest principal angle) between the subspaces $S_{\theta}$ and $S_{\theta'}$ \cite[Corollary 2.6]{Ste73}. Define the union
\begin{equation}
\label{eqn:defUnion}
\cU=\bigcup_{\theta \in \Theta}S_{\theta}.
\end{equation}
We are interested in reducing the dimensionality of $\cU$ using a subgaussian map. To achieve this, we apply Theorem~\ref{thm:KappaGeneral} with a suitable parametrization of $\cU\cap B_H$. We estimate the relevant $\ga_2$-functional in the following lemma.
\begin{lemma}
\label{lem:gamma2EstUnion}
Set $K=\sup_{\theta\in\Theta} \mathrm{dim}(S_{\theta})$. Let $\xi:\Theta\ti B_H\rightarrow \cU\cap B_H$ be the parametrization defined by $\xi(\theta,x) = P_{\theta}x$ and let $d_{\xi}$ be as in (\ref{eqn:dxiDef}). Then,
$$\ga_{2}(\Theta\ti B_H,d_{\xi}) \lesssim \sqrt{K} + \ga_{2}(\Theta,d_{\Fin}).$$
\end{lemma}
\begin{proof}
Let $(\Theta_n)_{n\geq 0}$ be any admissible sequence in $\Theta$ and let $\rho=(\rho_n)_{n\geq 0}$ be an associated sequence of maps $\rho_n:\Theta\rightarrow \Theta_n$. For any given $\theta \in \Theta$ and $n\geq 0$ we define a semi-metric on $B_H$ by
$$d_{n,\theta}(x,y) = \|P_{\rho_{n}(\theta)}(x-y)\|_{H}.$$
Next, for every $\theta \in \Theta$ we define an admissible sequence $\cH_{\theta} = (H_{n,\theta})_{n\geq 0}$ of $B_H$ by
$$H_{n,\theta} := \argmin_A \sup_{x\in B_H} d_{n,\theta}(x,A),$$
where the minimization is over all subsets $A$ of $B_H$ with $|A|\leq 2^{2^n}$. We use
$$e_{n,\theta} = \inf_{A} \sup_{x\in B_H} d_{n,\theta}(x,A) = \sup_{x\in B_H} d_{n,\theta}(x,H_{n,\theta})$$
to denote the associated entropy numbers. Finally, we define
$$\si_{n,\theta}(x) = \argmin_{y \in H_{n,\theta}} d_{n,\theta}(x,y).$$
For completeness we set $\Theta_{-1}$ equal to $\Theta_0$, $d_{-1,\theta}$ equal to $d_{0,\theta}$ and $H_{-1,\theta}$ equal to $H_{0,\theta}$. Now we define
$$T_n = \{(\rho_{n-1}(\theta),\si_{n-1,\theta}(x)) \in \Theta\ti B_H \ : \ \theta \in \Theta, x \in B_H\}.$$
Note that $\si_{n,\theta}(x)$ depends only on $\theta$ through $\rho_{n}(\theta)$. It follows that
$$|T_n| \leq |\Theta_{n-1}| \ 2^{2^{n-1}} \leq 2^{2^{n}}$$
and therefore $\cT=(T_n)_{n\geq 0}$ is an admissible sequence for $\Theta\ti B_H$. For $(\theta,x) \in \Theta\ti B_H$ we define $\pi_n(\theta,x) = (\rho_{n-1}(\theta),\si_{n-1,\theta}(x))$. By the triangle inequality,
\begin{align*}
d_{\xi}((\theta,x),\pi_{n}(\theta,x)) & = \|P_{\theta}x-P_{\rho_{n-1}(\theta)}\si_{n-1,\theta}(x)\|_{H} \\
& \leq \|(P_{\theta}-P_{\rho_{n-1}(\theta)})x\|_{H} + \|P_{\rho_{n-1}(\theta)}(x-\si_{n-1,\theta}(x))\|_{H} \\
& \leq \|P_{\theta} - P_{\rho_{n-1}(\theta)}\| + d_{n-1,\theta}(x,\si_{n-1,\theta}(x)) \\
& \leq \|P_{\theta} - P_{\rho_{n-1}(\theta)}\| + e_{n-1,\theta},
\end{align*}
where we used in the final estimate that,
$$d_{n,\theta}(x,\si_{n,\theta}(x)) = d_{n,\theta}(x,H_{n,\theta}) \leq \sup_{x\in B_H} d_{n,\theta}(x,H_{n,\theta}) = e_{n,\theta}.$$
Using these observations we obtain
\begin{align*}
\ga_2(\Theta\ti B_H,d_{\xi}) & \leq \sum_{n\geq 0} 2^{n/2} d_{\xi}((\theta,x),\pi_{n}(\theta,x)) \\
& \leq \sum_{n\geq 0}2^{n/2} d_{\Fin}(\theta,\rho_{n-1}(\theta)) + 2^{n/2} e_{n-1,\theta} \\
& \leq (1+\sqrt{2})\Big(\sum_{n\geq 0} 2^{n/2}d_{\Fin}(\theta,\rho_{n}(\theta)) + \sum_{n\geq 0} 2^{n/2}e_{n,\theta}\Big).
\end{align*}
It remains to bound the second term on the right hand side. Observe that
$$e_{n,\theta} = \inf\{u \ : \ N(B_H,d_{n,\theta},u)\leq 2^{2^n}\}.$$
If $(a_{\al})$ is a $u$-net for the unit ball in $S_{\rho_n(\theta)}$ with respect to $d_H$ and we pick $x_{\al}$ such that $a_{\al} = P_{\rho_n(\theta)}x_{\al}$, then $(x_{\al})$ is a $u$-net for $B_H$ with respect to $d_{n,\theta}$. Since $S_{\rho_n(\theta)}$ is at most $K$-dimensional, we find for all $u>0$,
$$N(B_H,d_{n,\theta},u) \leq N(B_{\R^K},d_2,u).$$
Thus we can conclude that $e_{n,\theta}\leq e_n$, where
$$e_{n} = \inf\{u \ : \ N(B_{\R^K},d_2,u)\leq 2^{2^n}\}.$$
Now, if $u<e_n$ then $N(B_{\R^K},d_2,u)\geq 2^{2^n}+1$ and hence we can estimate
\begin{align*}
\Big(1-\frac{1}{\sqrt{2}}\Big) \sum_{n\geq 0} 2^{n/2}e_{n} & \leq \sum_{n\geq 0} 2^{n/2}e_n - \sum_{n\geq 1} 2^{(n-1)/2}e_n \nonumber\\
& = \sum_{n\geq 0} 2^{n/2}(e_{n}-e_{n+1}) \nonumber\\
& \leq \frac{1}{\log^{1/2}(2)} \sum_{n\geq 0} \log^{1/2}(1+2^{2^n}) \ (e_{n}-e_{n+1}) \nonumber\\
& \leq \frac{1}{\log^{1/2}(2)} \sum_{n\geq 0} \int_{e_{n+1}}^{e_{n}} \log^{1/2} N(B_{\R^K},d_2,u) \ du \\
& = \frac{1}{\log^{1/2}(2)} \int_{0}^{1} \log^{1/2} N(B_{\R^K},d_2,u) \ du.
\end{align*}
As observed in Example~\ref{exa:covDimSubspace},
$$N(B_{\R^K},d_2,u) \leq (1+2u^{-1})^{K}.$$
Putting these estimates together we conclude using (\ref{eqn:intEstimate}) that
\begin{align*}
\sum_{n\geq 0} 2^{n/2}e_{n,\theta} & \leq \sqrt{K}\Big(\log^{1/2}(2) - \frac{\log^{1/2}(2)}{\sqrt{2}}\Big)^{-1}\int_0^{1} \log^{1/2}(1+2u^{-1}) \ du \\
& \leq \sqrt{K}\Big(\log^{1/2}(2) - \frac{\log^{1/2}(2)}{\sqrt{2}}\Big)^{-1}\log^{1/2}(3e).
\end{align*}
This completes the proof.
\end{proof}
Theorem~\ref{thm:KappaGeneral} and Lemma~\ref{lem:gamma2EstUnion} together imply the following result.
\begin{theorem}
\label{thm:RIPUnion}
Let $\cU$ be the union of subspaces defined in (\ref{eqn:defUnion}) and let $K=\sup_{\theta\in\Theta} \mathrm{dim}(S_{\theta})$. Let $\Phi:\Om\ti H\rightarrow \R^m$ be a subgaussian map on $\cU$. Then there is a constant $C>0$ such that for any $0<\del,\eta<1$ we have $\bP(\del_{\cU,\Phi}\geq \del)\leq\eta$ provided that
$$m \geq C\al^2\del^{-2} \ \max\{K+\ga_{2}^2(\Theta,d_{\mathrm{Fin}}), \log(\eta^{-1})\}.$$
\end{theorem}
\begin{proof}
Recall from (\ref{eqn:RIconstKappa}) that $\del_{\cU,\Phi}=\ka_{\cU_{nv},\Phi}$ and clearly $\ka_{\cU_{nv},\Phi}\leq \ka_{\cU\cap B_H,\Phi}$. Let $\xi$ be the parametrization of $\cU\cap B_H$ defined in Lemma~\ref{lem:gamma2EstUnion}. By Theorem~\ref{thm:KappaGeneral} we have $\bP(\ka_{\cU\cap B_H,\Phi}\geq \del)\leq\eta$ if
$$m \geq C\al^2\del^{-2} \ \max\{\ga_2^2(\Theta\ti B_H,d_{\xi}), \log(\eta^{-1})\}.$$
The assertion now follows from Lemma~\ref{lem:gamma2EstUnion}.
\end{proof}
Theorem~\ref{thm:RIPUnion} improves upon the second part of Corollary~\ref{cor:covDim} even if $\Theta$ is a finite set. Indeed, this follows from the bound $\ga_2^2(\Theta,d_{\Fin})\lesssim \log|\Theta|$.\par
Let us now derive a condition under which $\Phi$ preserves pairwise distances in $\cU$. For $\theta,\theta' \in \Theta$ let $S_{(\theta,\theta')}$ be the subspace spanned by $S_{\theta}$ and $S_{\theta'}$ and let $P_{(\theta,\theta')}$ be the projection onto this subspace.
\begin{theorem}
\label{thm:embedUnion}
Let $\Phi:\Om\ti H\rightarrow \R^m$ be a subgaussian map on $\cU$. Set $K=\sup_{\theta,\theta'}\mathrm{dim}(S_{(\theta,\theta')})$. On the set $\Theta\ti\Theta$ consider the metric
$$d_{\Fin}((\theta,\theta'),(\tau,\tau')) = \|P_{(\theta,\theta')} - P_{(\tau,\tau')}\|.$$
Then, there is a constant $C>0$ such that for any $0<\eps,\eta<1$ we have $\bP(\eps_{\cU,\Phi}\geq \eps)\leq\eta$ provided that
\begin{equation}
\label{eqn:condEmbedUnion}
m \geq C\al^2\eps^{-2} \ \max\{K + \ga_2^2(\Theta\ti\Theta,d_{\Fin}), \log(\eta^{-1})\}.
\end{equation}
\end{theorem}
\begin{proof}
Recall that $\eps_{\cU,\Phi}=\del_{\cU-\cU,\Phi}$. Since
$$\cU - \cU\subset \cU_* := \cup_{(\theta,\theta') \in \Theta\ti \Theta} S_{(\theta,\theta')},$$
we have $\del_{\cU-\cU,\Phi}\leq \del_{\cU_*,\Phi}$. The result follows by applying Theorem~\ref{thm:RIPUnion} to $\cU_*$, noting that $\mathrm{dim}(S_{(\theta,\theta')})\leq 2K$ for all $(\theta,\theta') \in S_{(\theta,\theta')}$.
\end{proof}
Clearly, if there exists a one-to-one map from $\cU$ into $\R^m$ then we must have $m\geq K$. In particular, the scaling of $m$ in $K$ in (\ref{eqn:condEmbedUnion}) cannot be improved.
\begin{remark}
\label{rem:Landweber}
Together with the main result of \cite{Blu11}, Theorem~\ref{thm:embedUnion} implies the following very general uniform signal recovery result. Suppose that we wish to recover a vector $x\in \cU$ from $m$ noisy measurements $y\in \R^m$ given by
\begin{equation}
\label{eqn:measurements}
y=\Phi x + e,
\end{equation}
where $e\in H$ represents the measurement error. If $m$ satisfies (\ref{eqn:condEmbedUnion}), then in the terminology of \cite{Blu11} the subgaussian map $\Phi$ is with probability $1-\eta$ a bilipschitz map on $\cU$ with constants $1-\eps$ and $1+\eps$. Therefore, if $(1+\eps)/(1-\eps)<3/2$, then with probability $1-\eta$ we can recover \emph{any} $x\in \cU$ robustly from the $m$ measurements $y$ in (\ref{eqn:measurements}) using a projective Landweber algorithm. We refer to \cite[Theorem 2]{Blu11} for details and a quantitative statement.
\end{remark}
In \cite{Man13,MaR13}, Mantzel and Romberg proved a version of Theorem~\ref{thm:RIPUnion} for a matrix $\Phi$ populated with i.i.d.\ standard Gaussian entries. They assume that $\Theta$ has covering dimension $K_{\Fin}$ with respect to $d_{\Fin}$, with base covering $N_0$. Their result in \cite{MaR13} states (in our terminology) that $\bP(\del_{\cU,\Phi}\geq \del)\leq \eta$ provided that
\begin{equation}
\label{eqn:lowerBoundMR}
m \geq C \max\{\del^{-1}\log(K),\del^{-2}\} \ \max\{K(K_{\Fin} + \log K + \log N_0),K\log(\eta^{-1})\}.
\end{equation}
Note that in this setup Theorem~\ref{thm:RIPUnion} implies that (cf.\ the argument in the proof of Corollary~\ref{cor:covDim})
\begin{equation}
\label{eqn:lowerBoundMRImp}
m \geq C\del^{-2}\max\{K + \log N_0 + K_{\Fin},\log(\eta^{-1})\}
\end{equation}
is already sufficient. Moreover, this statement extends to any subgaussian map, in particular the database-friendly map discussed in Example~\ref{exa:Achlioptas}.\par
The approach in \cite{Man13,MaR13} is very different from ours. The idea is to write
$$\del_{\cU,\Phi} = \sup_{\theta\in\Theta} \|P_{\theta}\Phi^*\Phi P_{\theta} - P_{\theta}\|$$
and to estimate the expected value of the right hand side using a (classical) chaining argument in the operator norm, based on the noncommutative Bernstein inequality. Note that this approach cannot yield the improved condition (\ref{eqn:lowerBoundMRImp}). For example, the factor $\log(K)$ in (\ref{eqn:lowerBoundMR}) is incurred through the use of the noncommutative Bernstein inequality, and therefore an artefact of the used method.

\section{Manifolds}
\label{sec:manifolds}

Let $\cM$ be a $K$-dimensional $C^1$-submanifold of $\R^n$, equipped with the Riemannian metric induced by the Euclidean inner product on $\R^n$. We use the following standard notation and terminology. For any $x\in \cM$ we let $T_x\cM$ denote the tangent space of $\cM$ at $x$ and let $P_x:\R^n\rightarrow T_x\cM$ be the associated projection onto $T_x\cM$. We use
$$T\cM = \bigcup_{x\in \cM} T_x\cM$$
to denote the tangent bundle of $\cM$. If $\ga:[a,b]\rightarrow \cM$ is a piecewise $C^1$ curve in $\cM$, then its length is defined as
$$L(\ga) = \int_a^b \|\ga'(t)\|_2 \ dt.$$
For $x,y\in \cM$, let $d_{\cM}(x,y)$ be the geodesic distance between $x$ and $y$, which can be described as
$$d_{\cM}(x,y) = \inf\{L(\ga) \ : \ \ga:[a,b]\rightarrow \cM \ \mathrm{piecewise} \ C^1, \ a,b\in\R, \ \ga(a)=x, \ \ga(b)=y\}.$$
For more information on Riemannian (sub)manifolds we refer to \cite{Lee97}.\par
Below we prove three different types of dimensionality reduction results for a subgaussian map $\Phi$. We derive a sufficient condition under which $\Phi$ uniformly preserves the lengths of all curves in $\cM$ up to a specified multiplicative error and conditions under which ambient distances are preserved up to an additive error and a multiplicative error, respectively.

\subsection{Preservation of curve lengths}

We can immediately apply Theorem~\ref{thm:RIPUnion} to derive a condition under which $\Phi$ \emph{uniformly} preserves the length of all curves in $\cM$ up to a specified multiplicative error.
\begin{theorem}
\label{thm:manifoldCurvePres}
Let $\cM$ be a $K$-dimensional $C^1$-submanifold of $\R^n$. Let $\Phi:\Om\ti\R^n\rightarrow \R^m$ be a subgaussian map. There is a constant $C>0$ such that for any $0<\eps,\eta<1$ we have with probability at least $1-\eta$ for any piecewise $C^1$-curve $\ga$ in $\cM$,
\begin{equation}
\label{eqn:curvePres}
(1-\eps)L(\ga)\leq L(\Phi\ga) \leq (1+\eps)L(\ga)
\end{equation}
provided that
\begin{equation}
\label{eqn:manifoldCurvePres}
m \geq C\al^2(2\eps-\eps^2)^{-2} \ \max\{K+\ga_{2}^2(\cM,d_{\mathrm{Fin}}), \log(\eta^{-1})\}.
\end{equation}
\end{theorem}
\begin{proof}
Let $\ga:[a,b]\rightarrow \cM$ be any piecewise $C^1$-curve in $\cM$, then $\Phi\ga$ is a piecewise $C^1$-curve and $(\Phi\ga)'(t)=\Phi\ga'(t)$ whenever $\ga$ is differentiable at $t$. Therefore,
$$(1-\del_{T\cM,\Phi})\|\ga'(t)\|_2^2 \leq \|(\Phi\ga)'(t)\|_2^2 \leq (1+\del_{T\cM,\Phi})\|\ga'(t)\|_2^2.$$
Note that if $\del_{T\cM,\Phi}\leq 2\eps-\eps^2$, then (see Remark~\ref{rem:noSquaresMult} for a similar observation)
$$(1-\eps)\|\ga'(t)\|_2 \leq \|(\Phi\ga)'(t)\|_2 \leq (1+\eps)\|\ga'(t)\|_2.$$
Integrating on both sides over $[a,b]$ yields
$$(1-\eps)L(\ga) \leq L(\Phi\ga) \leq (1+\eps)L(\ga).$$
By Theorem~\ref{thm:RIPUnion}, we have $\bP(\del_{T\cM,\Phi}\geq 2\eps-\eps^2)\leq \eta$ under condition (\ref{eqn:manifoldCurvePres}) and this implies the result.
\end{proof}
\begin{remark}
If the map $\Phi$ in Theorem~\ref{thm:manifoldCurvePres} also happens to be a manifold embedding, i.e., an immersion that is homeomorphic onto its image, then it preserves geodesic distances. That is, if for a given $0<\eps<1$, (\ref{eqn:curvePres}) holds for all piecewise $C^1$-curves in $\cM$, then
$$(1-\eps)d_{\cM}(x,y) \leq d_{\Phi\cM}(\Phi x,\Phi y) \leq (1+\eps)d_{\cM}(x,y), \qquad \mathrm{for \ all} \ x,y\in\cM.$$
Indeed, for given $x,y \in\cM$, let $\ga_{g,\Phi}$ be a geodesic between $\Phi(x)$ and $\Phi(y)$, and let $\ga$ be the preimage of $\ga_{g,\Phi}$. By (\ref{eqn:curvePres}),
$$(1-\eps)d_{\cM}(x,y)\leq (1-\eps)L(\ga) \leq L(\ga_{g,\Phi}) = d_{\Phi\cM}(\Phi(x),\Phi(y)).$$
Similarly, if $\ga_g$ is a geodesic between $x$ and $y$ in $\cM$, then
$$d_{\Phi\cM}(\Phi(x),\Phi(y)) \leq L(\Phi\ga_g) \leq (1+\eps)L(\ga_g) = (1+\eps)d_{\cM}(x,y).$$
\end{remark}

\subsection{Preservation of ambient distances: additive error}

We briefly consider maps that preserve pairwise ambient distances up to a specified additive error. The following result is similar to a result established for random projections in \cite[Theorem 9]{AHY13}.
\begin{proposition}
\label{pro:manifoldDD}
Let $\cM$ be a $C^1$-manifold with doubling dimension $D_{\cM}$ in the geodesic distance $d_{\cM}$ and let $\Del_{\cM}$ be its diameter in $d_{\cM}$. Let $\Phi:\Om\ti\R^n\rightarrow \R^m$ be a subgaussian map. Then, there is a constant $C>0$ such that $\bP(\zeta_{\cM,\Phi}\geq \zeta)\leq \eta$, provided that
$$m\geq C\al^2\zeta^{-2}\Del_{\cM}^4 \ \max\{D_{\cM},\log(\eta^{-1})\}.$$
\end{proposition}
\begin{proof}
Since $d_2\leq d_{\cM}$, we find using (\ref{eqn:gammaFunEstEntInt})
\begin{align*}
\ga_2(\cM,d_2) & \lesssim \int_0^{\Del_{\cM}} \log^{1/2} N(\cM,d_{\cM},\eps) \ d\eps \\
& = \Del_{\cM} \int_0^1 \log^{1/2} N(\cM,d_{\cM},\eps\Del_{\cM}) \ d\eps \\
& \leq \Del_{\cM} D_{\cM}^{1/2} \int_0^1 \log^{1/2}(c/\eps) \ d\eps \lesssim \Del_{\cM}D_{\cM}^{1/2},
\end{align*}
where in the final step we used (\ref{eqn:intEstimate}). The result is now immediate from the third statement in Corollary~\ref{cor:RIPepsIsom}.
\end{proof}
If $\ga$ is a $C^1$-curve in $\R^n$, then it has doubling dimension $2$ with respect to the geodesic distance. Therefore, Proposition~\ref{pro:manifoldDD} implies in this case that with probability $1-\eta$
$$\|x-y\|_2^2 - \zeta \leq \|\Phi(x-y)\|_2^2 \leq \|x-y\|_2^2 + \zeta, \qquad \mathrm{for \ all} \ x,y\in \ga,$$
whenever
$$m\geq C\al^2\zeta^{-2}\Del_{\ga}^4 \max\{2,\log(\eta^{-1})\}.$$

\subsection{Preservation of ambient distances: multiplicative error}

We will now investigate under which conditions a subgaussian map $\Phi$ on $\cM$ preserves pairwise ambient distances up to a small multiplicative error. These maps also approximately preserve several other properties of the manifold, such as its volume and the length and curvature of curves in the manifold (see \cite[Section 4.2]{BaW09} for a discussion). Results in this direction were first obtained in \cite{BaW09} and improved upon in \cite{Cla08,EfW13}.\par
Let us first observe an embedding result for manifolds with a low linearization dimension, which substantially improves \cite[Theorem 8]{AHY13}.
\begin{corollary}
For every $1\leq i\leq k$ let $\cM_i$ be smooth submanifold of $\R^n$ with linearization dimension $K_i$. Set $\cM=\cup_{i=1}^k \cM_i$ and $K=\max_i K_i$. Let $\Phi:\Om\ti\R^n\rightarrow \R^m$ be a subgaussian map. Then there is a constant $C>0$ such that for every $\eps,\eta>0$ we have $\bP(\eps_{\cM,\Phi}\geq \eps)\leq \eta$ if
\begin{equation*}
m \geq C\al^2\eps^{-2} \ \max\{\log k + K, \log(\eta^{-1})\}.
\end{equation*}
\end{corollary}
\begin{proof}
By \cite[Lemma 3]{AHY13}, $\cM_i$ is contained in an affine subspace of dimension $K_i$ and therefore in a linear subspace of dimension $K_i+1$. The result is now immediate from the second statement in Corollary~\ref{cor:covDim}.
\end{proof}
To derive the main results of this section, Theorems~\ref{thm:manifoldIota} and \ref{thm:manifoldReach}, we apply Corollary~\ref{cor:RIPepsIsom} and estimate the $\ga_2$-functional of the set $\cM_{nc}$ of normalized chords. We use (\ref{eqn:gammaFunEstEntInt}), i.e.,
$$\ga_2(\cM_{nc},d_2)\lesssim \int_0^1 \log^{1/2} N(\cM_{nc},d_2,u) \ du$$
and estimate the covering numbers of $\cM_{nc}$. The idea behind the covering number estimates, which is already implicit in \cite{Cla08}, is to divide the set of normalized chords into two categories. Firstly, one considers normalized chords corresponding to $x,y\in \cM$ which are `close' in the Euclidean metric. In \cite{Cla08}, these chords are called the `short chords', which should be taken as shorthand for `the normalized chords corresponding to short chords'. Let
$$\Ch(x,y) = \frac{y-x}{\|y-x\|_2}$$
denote the normalized chord from $x$ to $y$. Since $\Ch(x,y)$ converges to a unit tangent vector in $T_x\cM$ as $y$ approaches $x$, it is clear that this part of the covering number estimate requires good control of the `intrinsic dimension' of the tangent bundle of $\cM$. Secondly, one considers normalized chords corresponding to $x$ and $y$ which are `far apart' in Euclidean distance (the `long chords' in the terminology of \cite{Cla08}). These chords can be approximated well by chords $\Ch(a,b)$, where $a$ and $b$ are taken from a covering of $\cM$ itself, see Lemma~\ref{lem:longChords} below. To be able to decide whether two points are `close' or `far apart', we need to quantify how well we can approximate a normalized chord by a tangent vector. For this purpose we introduce the following parameter.
\begin{definition}
\label{def:iota}
If $\cM$ is a $C^1$-submanifold of $\R^n$, then we let $\iota(\cM)$ be the smallest constant $0<\iota\leq\infty$ satisfying
$$\|\Ch(x_1,x_2) - P_{x_1}\Ch(x_1,x_2)\|_2 \leq \iota\|x_1-x_2\|_2,\qquad \mathrm{for \ all} \ x_1,x_2\in\cM.$$
\end{definition}
In the proof of Theorem~\ref{thm:manifoldIota} we use the following observation, which is implicitly used in \cite{Cla08}. It is readily proven using the triangle and reverse triangle inequalities.
\begin{lemma}
\label{lem:longChords}
If $x_1,x_2 \in \R^n$ satisfy $\|x_1-x_2\|_2\geq t>0$, then
$$\|\Ch(x_1,x_2) - \Ch(y_1,y_2)\|_2 \leq 2t^{-1}(\|x_1-y_1\|_2 + \|x_2-y_2\|_2).$$
\end{lemma}
\begin{theorem}
\label{thm:manifoldIota}
Let $\cM$ be a $K$-dimensional $C^1$-submanifold of $\R^n$. Let $\Phi:\Om\ti\R^n\rightarrow \R^m$ be a subgaussian map. Suppose that $\cM$ has covering dimensions $K_2$ and $K_{\Fin}$ with respect to $d_2$ and $d_{\Fin}$, respectively. Then there is a constant $C>0$ such that for any $0<\eps,\eta<1$ we have $\bP(\eps_{\cM,\Phi}\geq \eps)\leq \eta$ provided that
$$m\geq C \al^2\eps^{-2} \ \max\{K_2\log_+(\iota(\cM)\Del_{d_2}(\cM)) + K_{\Fin} + K, \log(\eta^{-1})\}.$$
\end{theorem}
\begin{proof}
Let $0<a,b,c,t<\infty$ be parameters to be determined later. Let $N_2$ be an $a$-net of $\cM$ with respect to the Euclidean distance $d_2$ and let $N_{\Fin}$ be a $b$-net for $\cM$ with respect to $d_{\Fin}$. Finally, for any $y\in N_{\Fin}$ let $N_y$ be a $c$-net for the unit sphere $\cS_{\R^n}$ in $\R^n$ with respect to the induced semi-metric $d_y(z_1,z_2) := \|P_y(z_1-z_2)\|_2$.\par
Suppose first that $\|x_1-x_2\|_2>t$. Let $y_1,y_2\in N_2$ be such that $\|x_1-y_1\|_2<a$ and $\|x_2-y_2\|_2<a$. By Lemma~\ref{lem:longChords},
\begin{equation*}
\|\Ch(x_1,x_2) - \Ch(y_1,y_2)\|_2 \leq 2t^{-1}(\|x_1-y_2\|_2 + \|x_2-y_2\|_2) \leq 4t^{-1}a.
\end{equation*}
Suppose now that $\|x_1-x_2\|_2 \leq t$. Pick $y\in N_{\Fin}$ such that $\|P_{x_1}-P_y\|<b$ and subsequently $z\in N_y$ such that $\|P_y(\Ch(x_1,x_2)-z)\|_2<c$. Letting $\iota$ be as in Definition~\ref{def:iota}, we find
\begin{align*}
& \|\Ch(x_1,x_2) - P_yz\|_2 \\
& \qquad \leq \|\Ch(x_1,x_2) - P_{x_1}\Ch(x_1,x_2)\|_2 + \|P_{x_1}\Ch(x_1,x_2) - P_{y}\Ch(x_1,x_2)\|_2 \\
& \qquad \qquad \qquad + \|P_y(\Ch(x_1,x_2) - z)\|_2 \\
& \qquad \leq \iota t+b+c.
\end{align*}
Now let $0<u\leq 1$. From our estimates we see that if we pick $t=u/(3\iota)$, $a=u^2/(12\iota)$, $b=u/3$ and $c=u/3$, then
$$\{\Ch(y_1,y_2) \ : \ y_1,y_2 \in N_2\} \cup \{P_yz \ : \ y \in N_{\Fin}, \ z\in N_y\}$$
yields a $u$-net for $\cM_{nc}$ with respect to $d_2$. Since for every $y\in N_{\Fin}$ and $v>0$,
$$N(\cS_{\R^n},d_y,v) \leq N(B_{\R^K},d_2,v) \leq (1+2v^{-1})^K,$$
we obtain
\begin{equation}
\label{eqn:covNumDec}
N(\cM_{nc},d_2,u) \leq N^2\Big(\cM,d_2,\frac{u^2}{12\iota}\Big) + N\Big(\cM,d_{\Fin},\frac{u}{3}\Big)\Big(1+\frac{6}{u}\Big)^K.
\end{equation}
By (\ref{eqn:gammaFunEstEntInt}),
\begin{align}
\label{eqn:gammaMncEst}
\ga_2(\cM_{nc},d_2) & \lesssim \int_0^1 \log^{1/2} N(\cM_{nc},d_2,u) \ du \nonumber \\
& \leq 2\sqrt{2}\int_0^1 \log^{1/2} N\Big(\cM,d_2,\frac{u^2}{12\iota}\Big) \ du \nonumber\\
& \qquad + 2\int_0^1 \log^{1/2} N\Big(\cM,d_{\Fin},\frac{u}{3}\Big) \ du + 2\sqrt{K}\int_0^1 \log^{1/2}\Big(1+\frac{6}{u}\Big) \ du \nonumber\\
& \leq 2\sqrt{2 K_2}\int_0^1 \log_+^{1/2}\Big(\frac{12\iota\Del_{d_2}(\cM)}{u^2}\Big) \ du \nonumber\\
& \qquad + 2\sqrt{K_{\Fin}}\int_0^1 \log^{1/2}\Big(\frac{3}{u}\Big) \ du + \sqrt{K}\int_0^1 \log^{1/2}\Big(1+\frac{6}{u}\Big) \ du \nonumber\\
& \lesssim \sqrt{K_2}\log_+^{1/2}(\iota\Del_{d_2}(\cM)) + \sqrt{K_{\Fin}} + \sqrt{K},
\end{align}
where in the final step we used (\ref{eqn:intEstimate}). The result now follows from the second statement in Corollary~\ref{cor:RIPepsIsom}.
\end{proof}
We conclude by proving a result related to \cite{BaW09,EfW13} using some tools from \cite{EfW13}, which can in turn be traced back to \cite{NSW08}. Recall that the \emph{reach} $\tau(\cM)$ of a smooth submanifold $\cM$ of $\R^n$ is the smallest $\tau>0$ such that some point of $\R^n$ at distance $\tau$ from $\cM$ has two distinct points of $\cM$ as closest points in $\cM$.
\begin{lemma}
\label{lem:shortChords}
If $\cM$ has reach $\tau$, then $\iota(\cM)\leq 2\tau^{-1}$. Moreover, for any $x_1,x_2 \in \cM$,
$$d_{\Fin}(x_1,x_2) \leq 2\sqrt{2}\tau^{-1/2}\|x_1-x_2\|_2^{1/2}.$$
\end{lemma}
\begin{proof}
Suppose first that $\|x_1-x_2\|_2 \leq \tau/2$. Then,
\begin{align*}
\|\Ch(x_1,x_2) - P_{x_1}\Ch(x_1,x_2)\|_2 & = \sin(\angle(\Ch(x_1,x_2),P_{x_1}\Ch(x_1,x_2))) \\
& = \sin(\angle(x_2-x_1,P_{x_1}(x_2-x_1))) \leq \frac{\|x_1-x_2\|_2}{2\tau},
\end{align*}
where the final inequality follows from \cite[Lemma 2]{EfW13}. On the other hand, if $\|x_1-x_2\|_2>\tau/2$, then trivially,
$$\|\Ch(x_1,x_2) - P_{x_1}\Ch(x_1,x_2)\|_2\leq 1 \leq 2\tau^{-1}\|x_1-x_2\|_2.$$
The second statement for $x_1,x_2$ satisfying $\|x_1-x_2\|_2<\tau/2$ follows from \cite[Lemma 9]{EfW13} and is trivial if $\|x_1-x_2\|_2\geq\tau/2$.
\end{proof}
\begin{theorem}
\label{thm:manifoldReach}
Let $\cM$ be a $K$-dimensional $C^{\infty}$-submanifold of $\R^n$ with reach $\tau$ and covering dimension $K_2$ with respect to $d_2$. Let $\Phi:\Om\ti\R^n\rightarrow \R^m$ be a subgaussian map. Then there is a constant $C>0$ such that for any $0<\eps,\eta<1$ we have $\bP(\eps_{\cM,\Phi}\geq \eps)\leq \eta$ provided that
$$m\geq C \al^2\eps^{-2} \ \max\{K_2\log_+(\tau^{-1}\Del_{d_2}(\cM)) + K, \log(\eta^{-1})\}.$$
If $\cM$ has volume $V_{\cM}$, then $\bP(\eps_{\cM,\Phi}\geq \eps)\leq \eta$ if
$$m\geq C \al^2\eps^{-2} \ \max\{K\log_+(K\tau^{-1}) + K + \log_+(V_{\cM}), \log(\eta^{-1})\}.$$
\end{theorem}
\begin{proof}
By the second part of Lemma~\ref{lem:shortChords}, if $N$ is a $v$-net of $\cM$ with respect to $d_2$, then it is also a $2\sqrt{2}\tau^{-1/2}\sqrt{v}$-net with respect to $d_{\Fin}$. Hence,
$$N(\cM,d_{\Fin},2\sqrt{2}\tau^{-1/2}\sqrt{v}) \leq N(\cM,d_2,v),$$
which implies that for any $u>0$,
$$N(\cM,d_{\Fin},u)\leq N\Big(\cM,d_2,\frac{\tau}{8}u^2\Big).$$
Combining this with our estimate (\ref{eqn:covNumDec}) in the proof of Theorem~\ref{thm:manifoldIota}, we obtain
\begin{align}
\label{eqn:covNumEstReach}
N(\cM_{nc},d_2,u) & \leq N\Big(\cM,d_2,\frac{u^2}{12\iota}\Big) + N\Big(\cM,d_{\Fin},\frac{u}{3}\Big)\Big(1+\frac{6}{u}\Big)^K \nonumber\\
& \leq N\Big(\cM,d_2,\frac{\tau u^2}{24}\Big) + N\Big(\cM,d_2,\frac{\tau u^2}{72}\Big)\Big(1+\frac{6}{u}\Big)^K,
\end{align}
where we applied Lemma~\ref{lem:shortChords}. By a computation similar to (\ref{eqn:gammaMncEst}) we find
$$\ga_2(\cM_{nc},d_2) \lesssim \sqrt{K_2}\log_+^{1/2}(\tau^{-1}\Del_{d_2}(\cM)) + \sqrt{K}.$$
The first statement now follows from Corollary~\ref{cor:RIPepsIsom}.\par
For the second result we use that for any $v\leq \tau/2$ (cf.\ \cite[Lemma 11]{EfW13})
\begin{align*}
N(\cM,d_2,v) & \leq \Big(\frac{v^2}{4} - \frac{v^4}{64\tau^2}\Big)^{-K/2}V_{\cM}V_{B_{\R^K}}^{-1} \\
& \leq \Big(\frac{v^2}{4} - \frac{v^4}{64\tau^2}\Big)^{-K/2}\Big(\frac{K+2}{4\pi}\Big)^{K/2}V_{\cM}.
\end{align*}
We apply this bound to the terms on the far right hand side of (\ref{eqn:covNumEstReach}) to find for some absolute constants $c,\tilde{c}>0$ and $0<u\leq 1$,
$$N(\cM_{nc},d_2,u) \leq V_{\cM} \Big(\frac{K+2}{4\pi}\Big)^{K/2}\Big(\Big(\frac{c}{\tau u^2}\Big)^K + \Big(\frac{\tilde{c}}{\tau u^2}\Big)^K\Big(1+\frac{6}{u}\Big)^K\Big).$$
A computation similar to (\ref{eqn:gammaMncEst}) shows that
$$\ga_2(\cM_{nc},d_2) \lesssim \sqrt{K}(\log^{1/2}(K)+\log_+^{1/2}(\tau^{-1}) + 1) + \log_+^{1/2}(V_{\cM}).$$
The claim now follows from Corollary~\ref{cor:RIPepsIsom}.
\end{proof}
The first statement in Theorem~\ref{thm:manifoldReach} improves upon \cite[Theorem 3.1]{BaW09}. The second statement extends the result in
\cite[Theorem 2]{EfW13} from Gaussian matrices to general subgaussian maps and removes an additional $\cO(\log(\eps^{-1}))$ dependence of $m$ on $\eps$. The superfluous factor $\cO(\log(\eps^{-1}))$ seems to be an inherent construct of the proof in \cite{EfW13}, as it occurs in several other papers which use essentially the same method \cite{AHY13,Cla08,InN07}.

\section*{Acknowledgement}

The research for this paper was initiated after a discussion with Justin Romberg about compressive parameter estimation. I would like to thank him for providing me with William Mantzel's PhD thesis \cite{Man13}.


\begin{thebibliography}{10}

\bibitem{Ach03}
D.~Achlioptas.
\newblock Database-friendly random projections: {J}ohnson-{L}indenstrauss with
  binary coins.
\newblock {\em J. Comput. System Sci.}, 66(4):671--687, 2003.

\bibitem{AHY13}
P.~Agarwal, S.~Har-Peled, and H.~Yu.
\newblock Embeddings of surfaces, curves, and moving points in {E}uclidean
  space.
\newblock {\em SIAM J. Comput.}, 42(2):442--458, 2013.

\bibitem{Alo03}
N.~Alon.
\newblock Problems and results in extremal combinatorics. {I}.
\newblock {\em Discrete Math.}, 273(1-3):31--53, 2003.

\bibitem{BDD08}
R.~Baraniuk, M.~Davenport, R.~DeVore, and M.~Wakin.
\newblock A simple proof of the restricted isometry property for random
  matrices.
\newblock {\em Constr. Approx.}, 28(3):253--263, 2008.

\bibitem{BaW09}
R.~Baraniuk and M.~Wakin.
\newblock Random projections of smooth manifolds.
\newblock {\em Found. Comput. Math.}, 9(1):51--77, 2009.

\bibitem{BDL08}
G.~Biau, L.~Devroye, and G.~Lugosi.
\newblock On the performance of clustering in {H}ilbert spaces.
\newblock {\em IEEE Trans. Inform. Theory}, 54(2):781--790, 2008.

\bibitem{BiM01}
E.~Bingham and H.~Mannila.
\newblock Random projection in dimensionality reduction: applications to image
  and text data.
\newblock In {\em Proceedings of the seventh ACM SIGKDD international
  conference on Knowledge discovery and data mining}, pages 245--250. ACM,
  2001.

\bibitem{Blu11}
T.~Blumensath.
\newblock Sampling and reconstructing signals from a union of linear subspaces.
\newblock {\em IEEE Trans. Inform. Theory}, 57(7):4660--4671, 2011.

\bibitem{BlD09}
T.~Blumensath and M.~Davies.
\newblock Sampling theorems for signals from the union of finite-dimensional
  linear subspaces.
\newblock {\em IEEE Trans. Inform. Theory}, 55(4):1872--1882, 2009.

\bibitem{CaP11}
E.~Cand{\`e}s and Y.~Plan.
\newblock Tight oracle inequalities for low-rank matrix recovery from a minimal
  number of noisy random measurements.
\newblock {\em IEEE Trans. Inform. Theory}, 57(4):2342--2359, 2011.

\bibitem{CaT06}
E.~Cand{\`{e}}s and T.~Tao.
\newblock Near-optimal signal recovery from random projections: universal
  encoding strategies?
\newblock {\em IEEE Trans. Inform. Theory}, 52(12):5406--5425, 2006.

\bibitem{Cla08}
K.~Clarkson.
\newblock Tighter bounds for random projections of manifolds.
\newblock In {\em Proceedings of the twenty-fourth annual symposium on
  Computational geometry}, pages 39--48. ACM, 2008.

\bibitem{Das99}
S.~Dasgupta.
\newblock Learning mixtures of {G}aussians.
\newblock In {\em 40th {A}nnual {S}ymposium on {F}oundations of {C}omputer
  {S}cience ({N}ew {Y}ork, 1999)}, pages 634--644. IEEE Computer Soc., Los
  Alamitos, CA, 1999.

\bibitem{DaG03}
S.~Dasgupta and A.~Gupta.
\newblock An elementary proof of a theorem of {J}ohnson and {L}indenstrauss.
\newblock {\em Random Structures Algorithms}, 22(1):60--65, 2003.

\bibitem{Dir13}
S.~Dirksen.
\newblock Tail bounds via generic chaining.
\newblock ArXiv:1309.3522.

\bibitem{Don06}
D.~Donoho.
\newblock Compressed sensing.
\newblock {\em IEEE Trans. Inform. Theory}, 52(4):1289--1306, 2006.

\bibitem{DMM11}
P.~Drineas, M.~Mahoney, S.~Muthukrishnan, and T.~Sarl{\'o}s.
\newblock Faster least squares approximation.
\newblock {\em Numer. Math.}, 117(2):219--249, 2011.

\bibitem{EfW13}
A.~Eftekhari and M.~Wakin.
\newblock New analysis of manifold embeddings and signal recovery from
  compressive measurements.
\newblock ArXiv:1306.4748.

\bibitem{ElM09}
Y.~Eldar and M.~Mishali.
\newblock Robust recovery of signals from a structured union of subspaces.
\newblock {\em IEEE Trans. Inform. Theory}, 55(11):5302--5316, 2009.

\bibitem{FoR13}
S.~Foucart and H.~Rauhut.
\newblock {\em A Mathematical Introduction to Compressive Sensing}.
\newblock Birkha{\"{u}}ser, Boston, 2013.

\bibitem{FrM88}
P.~Frankl and H.~Maehara.
\newblock The {J}ohnson-{L}indenstrauss lemma and the sphericity of some
  graphs.
\newblock {\em J. Combin. Theory Ser. B}, 44(3):355--362, 1988.

\bibitem{GNE13}
R.~Giryes, S.~Nam, M.~Elad, R.~Gribonval, and M.~Davies.
\newblock Greedy-like algorithms for the cosparse analysis model.
\newblock {\em To appear in Linear Algebra and its Applications}, 2013.

\bibitem{Gor88}
Y.~Gordon.
\newblock On {M}ilman's inequality and random subspaces which escape through a
  mesh in {${\bf R}\sp n$}.
\newblock In {\em Geometric aspects of functional analysis (1986/87)}, volume
  1317 of {\em Lecture Notes in Math.}, pages 84--106. Springer, Berlin, 1988.

\bibitem{BHW07}
C.~Hegde, M.~Wakin, and R.~Baraniuk.
\newblock Random projections for manifold learning.
\newblock In {\em Advances in neural information processing systems}, pages
  641--648, 2007.

\bibitem{InM98}
P.~Indyk and R.~Motwani.
\newblock Approximate nearest neighbors: Towards removing the curse of
  dimensionality.
\newblock In {\em Proceedings of the Thirtieth Annual ACM Symposium on Theory
  of Computing}, STOC '98, pages 604--613, New York, NY, USA, 1998. ACM.

\bibitem{InN07}
P.~Indyk and A.~Naor.
\newblock Nearest-neighbor-preserving embeddings.
\newblock {\em ACM Trans. Algorithms}, 3(3):Art. 31, 12, 2007.

\bibitem{JaW13}
T.~Jayram and D.~Woodruff.
\newblock Optimal bounds for {J}ohnson-{L}indenstrauss transforms and streaming
  problems with subconstant error.
\newblock {\em ACM Trans. Algorithms}, 9(3):Art. 26, 17, 2013.

\bibitem{JoL84}
W.~Johnson and J.~Lindenstrauss.
\newblock Extensions of {L}ipschitz mappings into a {H}ilbert space.
\newblock In {\em Conference in modern analysis and probability ({N}ew {H}aven,
  {C}onn., 1982)}, volume~26 of {\em Contemp. Math.}, pages 189--206. Amer.
  Math. Soc., Providence, RI, 1984.

\bibitem{KMV12}
A.~Kalai, A.~Moitra, and G.~Valiant.
\newblock Disentangling gaussians.
\newblock {\em Commun. ACM}, 55(2):113--120, February 2012.

\bibitem{KlM05}
B.~Klartag and S.~Mendelson.
\newblock Empirical processes and random projections.
\newblock {\em J. Funct. Anal.}, 225(1):229--245, 2005.

\bibitem{Lee97}
J.~Lee.
\newblock {\em Riemannian manifolds}.
\newblock Springer-Verlag, New York, 1997.

\bibitem{LuD08}
Y.~Lu and M.~Do.
\newblock A theory for sampling signals from a union of subspaces.
\newblock {\em IEEE Trans. Signal Process.}, 56(6):2334--2345, 2008.

\bibitem{MaM12}
O.~Maillard and R.~Munos.
\newblock Linear regression with random projections.
\newblock {\em J. Mach. Learn. Res.}, 13:2735--2772, 2012.

\bibitem{Man13}
W.~Mantzel.
\newblock {\em Parametric estimation of randomly compressed functions}.
\newblock Ph{D}. {T}hesis, {G}eorgia {I}nsitute of {T}echnology, 2013.

\bibitem{MaR13}
W.~Mantzel and J.~Romberg.
\newblock Compressive parameter estimation.
\newblock Preprint, 2013.

\bibitem{MRS12}
W.~Mantzel, J.~Romberg, and K.~Sabra.
\newblock Compressive matched-field processing.
\newblock {\em The Journal of the Acoustical Society of America}, 132(1), 2012.

\bibitem{Mat08}
J.~Matou{\v{s}}ek.
\newblock On variants of the {J}ohnson-{L}indenstrauss lemma.
\newblock {\em Random Structures Algorithms}, 33(2):142--156, 2008.

\bibitem{MPT07}
S.~Mendelson, A.~Pajor, and N.~Tomczak-Jaegermann.
\newblock Reconstruction and subgaussian operators in asymptotic geometric
  analysis.
\newblock {\em Geom. Funct. Anal.}, 17(4):1248--1282, 2007.

\bibitem{MPT08}
S.~Mendelson, A.~Pajor, and N.~Tomczak-Jaegermann.
\newblock Uniform uncertainty principle for {B}ernoulli and subgaussian
  ensembles.
\newblock {\em Constr. Approx.}, 28(3):277--289, 2008.

\bibitem{NDE13}
S.~Nam, M.~Davies, M.~Elad, and R.~Gribonval.
\newblock The cosparse analysis model and algorithms.
\newblock {\em Appl. Comput. Harmon. Anal.}, 34(1):30--56, 2013.

\bibitem{NSW08}
P.~Niyogi, S.~Smale, and S.~Weinberger.
\newblock Finding the homology of submanifolds with high confidence from random
  samples.
\newblock {\em Discrete Comput. Geom.}, 39(1-3):419--441, 2008.

\bibitem{RWH10}
M.~Raginsky, R.~Willett, Z.~Harmany, and R.~Marcia.
\newblock Compressed sensing performance bounds under {P}oisson noise.
\newblock {\em IEEE Trans. Signal Process.}, 58(8):3990--4002, 2010.

\bibitem{RSS13}
H.~Rauhut, R.~Schneider, and Z.~Stojanac.
\newblock Low-rank tensor recovery via iterative hard thresholding.
\newblock Preprint.

\bibitem{RFP10}
B.~Recht, M.~Fazel, and P.~Parrilo.
\newblock Guaranteed minimum-rank solutions of linear matrix equations via
  nuclear norm minimization.
\newblock {\em SIAM Rev.}, 52(3):471--501, 2010.

\bibitem{Sar06}
T.~Sarl{\'{o}}s.
\newblock Improved approximation algorithms for large matrices via random
  projections.
\newblock In {\em Foundations of Computer Science, 2006. FOCS '06. 47th Annual
  IEEE Symposium on}, pages 143--152, Oct 2006.

\bibitem{Sch00}
L.~Schulman.
\newblock Clustering for edge-cost minimization.
\newblock In {\em Proceedings of the Thirty-second Annual ACM Symposium on
  Theory of Computing}, STOC '00, pages 547--555, New York, NY, USA, 2000. ACM.

\bibitem{Ste73}
G.~Stewart.
\newblock Error and perturbation bounds for subspaces associated with certain
  eigenvalue problems.
\newblock {\em SIAM Rev.}, 15:727--764, 1973.

\bibitem{Tal87}
M.~Talagrand.
\newblock Regularity of {G}aussian processes.
\newblock {\em Acta Math.}, 159(1-2):99--149, 1987.

\bibitem{Tal01}
M.~Talagrand.
\newblock Majorizing measures without measures.
\newblock {\em Ann. Probab.}, 29(1):411--417, 2001.

\bibitem{Tal05}
M.~Talagrand.
\newblock {\em The generic chaining}.
\newblock Springer-Verlag, Berlin, 2005.

\bibitem{Vem04}
S.~Vempala.
\newblock {\em The random projection method}.
\newblock DIMACS Series in Discrete Mathematics and Theoretical Computer
  Science, 65. American Mathematical Society, Providence, RI, 2004.

\bibitem{Ver13}
N.~Verma.
\newblock Distance preserving embeddings for general {$n$}-dimensional
  manifolds.
\newblock Preprint, 2013.

\end{thebibliography}
\end{document}